\documentclass[sigplan]{acmart}

\acmConference[~]{~}{~}
\setcopyright{none}
\usepackage{booktabs}   %% For formal tables:
\usepackage{subcaption} %% For complex figures with subfigures/subcaptions
\usepackage{microtype}

\usepackage[author=anonymous,nomargin,marginclue,footnote,status=draft]{fixme}
\FXRegisterAuthor{ls}{als}{LS}
\FXRegisterAuthor{dp}{atm}{dp}
\FXRegisterAuthor{bk}{abk}{bk}
\FXRegisterAuthor{pw}{apw}{PW}
\usepackage{amsmath}
\usepackage{amsthm}
\usepackage{stmaryrd}
\usepackage{bussproofs}
\usepackage[all]{xy}
\usepackage{amsfonts}
\usepackage{mathtools}

\newcounter{blubber}

\newenvironment{sparitemize}
{\begin{list}{\labelitemi}{
    \setlength{\leftmargin}{0pt}
    \setlength{\parsep}{0pt}
    \setlength{\itemindent}{4ex}
    \setlength{\itemsep}{0pt}
  }
}{\end{list}}

\newcommand{\ST}{\mathsf{ST}}
\newcommand{\ALC}{\mathcal{ALC}}

\newcommand{\Rat}{\mathbb{Q}}

\newcommand{\At}{\mathsf{At}}

\newcommand{\Nat}{\mathbb{N}}

\newcommand{\Set}{\mathsf{Set}}

\newcommand{\nonexp}[2]{#1 \to_1 #2}
\newcommand{\supnorm}[1]{\lVert #1 \rVert_\infty}
\newcommand{\nbhood}[2]{{U^{#1}(#2)}}
\newcommand{\CA}{\mathcal{A}}
\newcommand{\CB}{\mathcal{B}}
\newcommand{\CC}{\mathcal{C}}
\newcommand{\qr}{\mathsf{qr}}
\newcommand{\rk}{\mathsf{rk}}
\newcommand{\ffun}{\mathsf{F}}
\newcommand{\CF}{\mathcal{F}}
\newcommand{\mfun}{\mathsf{G}}
\newcommand{\ball}[3]{B_{#2}({#3})}
\newcommand{\modf}[1]{\mathcal{L}_{#1}}
\newcommand{\bang}{!}
\newcommand{\rnbhood}[3]{{{#1}^{#2}_{#3}}}
\newcommand{\wbar}[1]{\bar #1}%{\mkern 1.5mu\overline{\mkern-1.5mu#1\mkern-1.5mu}\mkern 1.5mu}
\newcommand{\ineqp}[2]{#2 \le #1 \le #2 + \tfrac{\epsilon}{2}}
\newcommand{\ineqm}[2]{#2 - \tfrac{\epsilon}{2} \le #1 \le #2}
\newcommand{\ineqb}[2]{#2 - \tfrac{\epsilon}{2} \le #1 \le #2 + \tfrac{\epsilon}{2}}

\theoremstyle{plain}
\newtheorem{thm}{Theorem}[section]
\newtheorem{lem}[thm]{Lemma}

\newtheorem{cor}[thm]{Corollary}

\theoremstyle{definition}
\newtheorem{defn}[thm]{Definition}
\newtheorem{expl}[thm]{Example}

\newtheorem{rem}[thm]{Remark}

\newtheorem{assn}[thm]{Assumption}

\begin{document}

%% Title information
\title[A van Benthem Theorem for Fuzzy Modal Logic]{A van Benthem Theorem for Fuzzy Modal Logic}
%     [Short Title]{Full Title}         %% [Short Title] is optional;
                                        %% when present, will be used in
                                        %% header instead of Full Title.
%\titlenote{with title note}             %% \titlenote is optional;
                                        %% can be repeated if necessary;
                                        %% contents suppressed with 'anonymous'
%\subtitle{Subtitle}                     %% \subtitle is optional
%\subtitlenote{with subtitle note}       %% \subtitlenote is optional;
                                        %% can be repeated if necessary;
                                        %% contents suppressed with 'anonymous'

%% Author information
%% Contents and number of authors suppressed with 'anonymous'.
%% Each author should be introduced by \author, followed by
%% \authornote (optional), \orcid (optional), \affiliation, and
%% \email.
%% An author may have multiple affiliations and/or emails; repeat the
%% appropriate command.
%% Many elements are not rendered, but should be provided for metadata
%% extraction tools.

%% Author with single affiliation.
\author{Paul Wild}
\author{Lutz Schr\"oder}
%\authornote{with author1 note}          %% \authornote is optional;
                                        %% can be repeated if necessary
%\orcid{nnnn-nnnn-nnnn-nnnn}             %% \orcid is optional
\affiliation{
  %\position{Position1}
  %\department{Department1}              %% \department is recommended
  \institution{Friedrich-Alexander-Universit\"at
  Erlangen-N\"urnberg}            %% \institution is required
  %\streetaddress{Street1 Address1}
  %\city{City1}
  %\state{State1}
  %\postcode{Post-Code1}
  %\country{Country1}                    %% \country is recommended
}
%\email{paul.wild,lutz.schroeder@fau.de}          %% \email is recommended

%% Author with two affiliations and emails.
\author{Dirk Pattinson}
%\authornote{with author2 note}          %% \authornote is optional;
                                        %% can be repeated if necessary
%\orcid{nnnn-nnnn-nnnn-nnnn}             %% \orcid is optional
\affiliation{
  %\position{Position2a}
  %\department{Department2a}             %% \department is recommended
  \institution{Australian National University, Canberra}           %% \institution is required
  %\streetaddress{Street2a Address2a}
  %\city{City2a}
  %\state{State2a}
  %\postcode{Post-Code2a}
  %\country{Country2a}                   %% \country is recommended
}
%\email{dirk.pattinson@anu.edu.au}         %% \email is recommended
\author{Barbara K\"onig}
\affiliation{
%  \position{Position2b}
%  \department{Department2b}             %% \department is recommended
  \institution{Universit\"at Duisburg-Essen}           %% \institution is required
%  \streetaddress{Street3b Address2b}
%  \city{City2b}
%  \state{State2b}
%  \postcode{Post-Code2b}
%  \country{Country2b}                   %% \country is recommended
}
%\email{barbara_koenig@uni-due.de}         %% \email is recommended

%% Abstract
%% Note: \begin{abstract}...\end{abstract} environment must come
%% before \maketitle command
\begin{abstract} 
  We present a fuzzy (or quantitative) version of the van
  Benthem theorem, which characterizes propositional modal logic
  as the bisimulation-invariant fragment of first-order
  logic. Specifically, we consider a first-order fuzzy predicate logic
  along with its modal fragment, and show that the fuzzy first-order
  formulas that are non-expansive w.r.t.\ the natural notion of
  bisimulation distance are exactly those that can be approximated by
  fuzzy modal formulas.
\end{abstract}


%% 2012 ACM Computing Classification System (CSS) concepts
%% Generate at 'http://dl.acm.org/ccs/ccs.cfm'.
% <concept_id>10011007.10011006.10011008</concept_id>
% <concept_desc>Software and its engineering~General programming languages</concept_desc>
% <concept_id>10003456.10003457.10003521.10003525</concept_id>
% <concept_desc>Social and professional topics~History of programming languages</concept_desc>
% <concept_significance>300</concept_significance>

% \ccsdesc[500]{Software and its engineering~General programming languages}
% \ccsdesc[300]{Social and professional topics~History of programming languages}
%% End of generated code

%% Keywords
%% comma separated list
 %% \keywords are mandatory in final camera-ready submission

%% \maketitle
%% Note: \maketitle command must come after title commands, author
%% commands, abstract environment, Computing Classification System
%% environment and commands, and keywords command.
\maketitle

\section{Introduction}
  \label{sec:intro}
  Fuzzy logic is a form of multi-valued logic originally studied by
  \L{}ukasiewicz and Tarski~\cite{LukasiewiczTarski30} and later
  popularized as a logic of \emph{vagueness} by
  Zadeh~\cite{Zadeh65}. It is based on replacing the standard set of
  Boolean truth values with a different lattice, most often, like in
  the present paper, the unit interval. Saying that a formula~$\phi$
  has truth value~$r\in[0,1]$ then means that~$\phi$ holds with
  \emph{degree}~$r$, which would apply to typical vague qualifications
  such as a given person being tall (in contrast to assigning a
  \emph{probability} $p\in[0,1]$ to~$\phi$, which would be read as
  saying that~$\phi$ is either completely true with probability~$p$ or
  completely false with probability $1-p$, as in `the die under the
  cup shows a~$3$ with probability~$p$').

  Beyond the original propositional setup, fuzzy truth values appear
  in variants of more expressive logics, notably in \emph{fuzzy
    first-order logics}~\cite{ChangKeisler66,Novak87,Hajek98} and in
  various \emph{fuzzy modal logics}. The latter go back to many-valued
  modal logics based on making valuations in Kripke
  models~\cite{RosserTurquette52,Segerberg67,Schotch76,Morgan79,Ostermann88,Morikawa88}
  or additionally also the accessibility
  relation~\cite{Fitting91} many-valued, and are nowadays
  maybe most popular in their incarnation as fuzzy description logics
  (e.g.~\cite{Yen91,TrespMolitor98,Straccia98,SanchezTettamanzi05,Hajek05};
  see~\cite{LukasiewiczStraccia08} for an overview). Many-valued modal
  fixpoint logics are also used in software model checking
  (e.g.~\cite{BrunsGodefroid04,KupfermanLustig07}).

  Like in the classical case, fuzzy modal logics typically embed into
  their first-order counterparts. In the classical setting, the core
  result on this embedding is \emph{van Benthem's theorem}, which
  states that a first-order formula~$\phi$ is equivalent to a modal
  formula if and only if~$\phi$ is invariant under
  bisimulation~\cite{BenthemThesis}. This is a form of expressive
  completeness: Modal logic expresses only bisimulation-invariant
  properties, but for such properties it is as expressive as
  first-order logic. Briefly, the aim of the current paper is to
  provide a counterpart of this theorem for a fuzzy modal logic.
  
  There is a wide variety of possible semantics for the fuzzy
  propositional connectives (see~\cite{LukasiewiczStraccia08} for an
  overview), employing, e.g., the additive structure
  (\emph{\L{}ukasiewicz logic}), the multiplicative structure
  (\emph{product logic}) or the Heyting algebra structure
  (\emph{G\"odel logic}) of the unit interval. For technical reasons,
  we work with the simplest possible semantics where conjunction is
  interpreted as minimum and all other connectives are derived using
  the classical encodings, effectively a fragment of \L{}ukasiewicz
  logic often called \emph{Zadeh logic}. That is, we consider
  \emph{Zadeh fuzzy modal logic}, more precisely \emph{Zadeh fuzzy
    $K$} or in description logic terminology \emph{Zadeh fuzzy
    $\ALC$}~\cite{Straccia98}, with \emph{Zadeh fuzzy first-order
    logic} as the first-order correspondence language, essentially the
  Zadeh fragment of Novak's \L{}ukasiewicz fuzzy first order
  logic~\cite{Novak87}.

  It has long been recognized that for quantitative systems, notions
  of \emph{behavioural distance} are more natural than two-valued
  bisimilarity~\cite{BreugelWorrell05}. In such a metric setting,
  bisimulation invariance becomes non-expansivity w.r.t.\ behavioural
  distance (e.g.\ if one views classical bisimilarity as a
  $\{0,1\}$-valued pseudometric, then non-expansivity means that
  distance~$0$ is preserved, which is precisely bisimulation
  invariance). The first step in our program is therefore to establish
  a notion of behavioural distance for fuzzy relational systems. We
  consider three different ways to define such a behavioural metric:
  via the modal logic, via a bisimulation game (similarly as in work
  on probabilistic systems~\cite{DesharnaisEA08}), or via a fixpoint
  characterization based on the Kantorovich lifting (similarly as
  in~\cite{bbkk:behavioral-metrics-functor-lifting}).  We show that
  they all coincide; in particular we obtain a Hennessy-Milner type
  theorem (behavioural distance equals logical distance).  This gives
  us a stable notion of behavioural metric for fuzzy relational
  systems.

  Our main result then says that \emph{the fuzzy modal formulas lie
    dense in the bisimulation-invariant first-order formulas}, where
  by bisimulation-invariant we now mean non-expansive w.r.t.\
  behavioural distance. In other words, every bisimulation-invariant
  fuzzy first-order formula can be modally approximated. The proof
  follows a strategy introduced for the classical case by
  Otto~\cite{o:van-Benthem-Rosen-elementary}, going via locality
  w.r.t.\ an adapted notion of Gaifman distance to show that every
  bisimulation-invariant fuzzy first-order formula is already
  non-expansive w.r.t.\ \emph{depth-$k$} behavioural distance for
  some~$k$ (this distance arises, e.g., by limiting the bisimulation
  game to~$k$~rounds). The key part of our technical development is,
  then, to establish a fuzzy counterpart of what in the classical case
  is a triviality: The classical proof ends in remarking that
  \emph{every} state property (without any assumption of first-order
  definability) of relational transition systems that is invariant
  under depth-$k$ bisimilarity is expressible by a modal formula of
  modal rank~$k$. In the fuzzy setting, this becomes a non-trivial
  result of independent interest: \emph{The fuzzy modal formulas of
    modal rank~$k$ lie dense in the fuzzy state properties that are
    non-expansive w.r.t.\ depth-$k$ behavioural distance.}

  Proofs are mostly omitted or only sketched; full proofs are in the
  appendix.
\paragraph*{Related Work}

Van Benthem's theorem was later shown by Rosen~\cite{Rosen97} to hold
also over finite structures. Modal characterization theorems have
since been proved in various settings, e.g.\ logics with frame
conditions~\cite{DawarOtto05}, coalgebraic modal
logics~\cite{SchroderEA17}, fragments of
XPath~\cite{tenCateEA10,FigueiraEA15,AbriolaEA17}, neighbourhood
logic~\cite{HansenEA09}, modal logic with team
semantics~\cite{KontinenEA15}, modal $\mu$-calculi (within monadic
second order logics)~\cite{JaninWalukiewicz95,EnqvistEA15}, PDL
(within weak chain logic)~\cite{Carreiro15}, modal first-order
logics~\cite{Benthem01,SturmWolter01}, and two-dimensional modal
logics with an $S5$-modality~\cite{WildSchroder17}. All these results
concern two-valued logics; we are not aware of any previous work of
this type for fuzzy modal logics.

There is, however, work on behavioural distances and fuzzy bisimulation
in connection with fuzzy  modal logic. We discuss only fuzzy notions
of bisimulation, omitting work on classical behavioural equivalence for
fuzzy transition systems and fuzzy automata. Balle et
al.~\cite{bgp:metrics-weighted-automata} consider bisimulation
metrics for weighted automata in order to characterize approximate
minimization. Cao et al. \cite{cswc:beh-dist-fuzzy-ts} study a notion
of behavioural distance for fuzzy transition systems, where the
lifting of the metric is derived from a transportation problem (the
dual of the Kantorovich metric), but without considering modal
logics. Fan~\cite{Fan15} proves a Hennessy-Milner type theorem for a
fuzzy modal logic with G\"odel semantics and a notion of fuzzy
bisimilarity. In \cite{fl:regular-equ-social-networks} she considers
an application to social network analysis and also observes that
\L{}ukasiewicz logic is problematic in this context (since the
operators do not preserve non-expansivity). Eleftheriou et
al. \cite{ekn:bisim-heyting} show a Hennessy-Milner theorem for
Heyting-valued modal logics as introduced by
Fitting~\cite{Fitting91}.

While we work in a fuzzy setting, we were inspired by related work on
probabilistic systems: Desharnais et al. studied behavioural distances
on logics \cite{dgjp:metrics-labelled-markov} as well as a game
characterization of probabilistic bisimulation
\cite{DesharnaisEA08}. A Hennessy-Milner theorem for the probabilistic
case is presented in \cite{BreugelWorrell05}, based on a coalgebraic
semantics.

\section{Fuzzy Modal Logic}
  \label{sec:logics}
  We proceed to recall the syntax and semantics of \emph{Zadeh
    fuzzy~$K$} or equivalently \emph{Zadeh
    fuzzy~$\ALC$}~\cite{Straccia98}, along with its first-order
  correspondence language. For simplicity we restrict the exposition
  to the unimodal case; the development extends straightforwardly to
  the multimodal case by just adding more indices. Formulas
  $\phi,\psi$ of \emph{fuzzy modal logic} are given by grammar
  \begin{equation*}
    \phi,\psi :: = c\mid p\mid \phi\ominus c\mid\neg\phi\mid\phi\land\psi\mid\Diamond\phi
  \end{equation*}
  where $p$ ranges over a fixed set $\At$ of \emph{propositional
    atoms} and~$c\in\Rat\cap[0,1]$ over rational truth constants. The
  syntax is thus mostly the same as for standard modal logic; the only
  additional ingredients are the truth constants and \emph{modified
    subtraction}~$\ominus$ as used in real-valued modal logics for
  probabilistic systems~\cite{BreugelWorrell05}. Further logical
  connectives are defined by the classical encodings, e.g.\
  $\phi\lor\psi$ abbreviates $\neg(\neg\phi\land\neg\psi)$, and
  $\phi\to\psi$ abbreviates $\neg\phi\lor\psi$; also, we introduce a
  dual modality $\Box$ as $\Box\phi:=\neg\Diamond\neg\phi$. The
  \emph{rank} $\rk(\phi)$ of a formula~$\phi$ is the maximal nesting
  depth of the modality~$\Diamond$ \emph{and propositional atoms}
  in~$\phi$. Formally, $\rk(\phi)$ is thus defined recursively by
  $\rk(c)=0$, $\rk(p)=1$, $\rk(\Diamond\phi)=1+\rk(\phi)$, and obvious
  clauses for the remaining constructs. We write $\modf{k}$ for the
  set of modal formulas of rank at most~$k$.

The \emph{semantics} of the logic is defined over \emph{fuzzy
  relational models} (or just \emph{models})
\begin{equation*}
  \CA = (A,(p^\CA)_{p\in\At},R^\CA)
\end{equation*}
consisting of a set $A$ of \emph{states}, a map $p^\CA: A \to [0,1]$
for each $p \in \At$, and a map $R^\CA: A \times A \to [0,1]$; we will
drop superscripts~$\CA$ when clear from the context. That is,
propositional atoms are interpreted as fuzzy predicates on the state
set, and states are connected by a binary fuzzy transition relation,
where \emph{fuzzy} is short for $[0,1]$-valued (as usual, we use
\emph{crisp} as an informal opposite of fuzzy, i.e.\ crisp means
two-valued). Fuzzy relational models are a natural fuzzification of
Kripke models, and in fact the instantiation of latticed Kripke models
over DeMorgan lattices~\cite{BrunsGodefroid04,KupfermanLustig07} to
the lattice~$[0,1]$; they arise from \emph{fuzzy transition systems}
(e.g.~\cite{cswc:beh-dist-fuzzy-ts}; \emph{fuzzy automata} go back as
far as~\cite{WeeFu69}) by adding propositional atoms. Unless stated
otherwise, we adhere to the convention that models are denoted by
calligraphic letters and their state sets by the corresponding italic.

We use $\land$, $\lor$ to denote meets and joins in $[0,1]$. A modal
formula $\phi$ is then assigned a fuzzy truth value $\phi_\CA(a)$, or
just $\phi(a)$, at every state $a\in A$, defined inductively by
\begin{align*}
    c(a) & = c \qquad p(a)  = p^\CA(a) \\
  (\phi\ominus c)(a) & = \max(\phi(a)- c,0)\\
  (\neg\phi)(a) & = 1-\phi(a)\\
  (\phi\land\psi)(a) & = \phi(a)\land\psi(a)\\
  (\Diamond\phi)(a)&=\textstyle\bigvee_{a'\in A}(R^\CA(a,a')\land\phi(a')).
\end{align*}
\noindent For brevity, we often conflate formulas and their evaluation
functions in both notation and vernacular, e.g.\ in statements
claiming that certain modal formulas form a dense subset of some set
of state properties.
\begin{rem}\label{rem:zadeh}
  As indicated above, we thus equip the propositional connectives with
  Zadeh semantics. This corresponds to widespread usage but is not
  without disadvantages in comparison to \L{}ukasiewicz semantics,
  which defines the conjunction of $a,b\in[0,1]$ as $\max(a+b-1,0)$;
  e.g.\ implication is the residual of conjunction in \L{}ukasiewicz
  semantics but not in Zadeh semantics (see~\cite{KunduChen98} for a
  more detailed discussion). We will later point out where this choice
  becomes most relevant; roughly speaking, \L{}ukasiewicz semantics is
  not easily reconciled with behavioural distance.

  The modal syntax as given above is essentially identical to the one
  used by van Breughel and Worrell to characterize behavioural
  distance in probabilistic transition
  systems~\cite{BreugelWorrell05}. Semantically, fuzzy models differ
  from probabilistic ones in that they do not require truth values of
  successor edges to sum up to~$1$, and moreover in the probabilistic
  setting the modality~$\Diamond$ is interpreted by expected truth
  values instead of suprema. The semantics of the propositional
  connectives, on the other hand, is in fact the same in both cases.
\end{rem}
\begin{expl}
  We can see fuzzy~$K$ as a logic of fuzzy transition systems
  (e.g.~\cite{cswc:beh-dist-fuzzy-ts}). E.g.\ the formula
  $\Diamond\Box 0$ then describes, roughly speaking, the degree to
  which a deadlocked state can be reached in one step. Formally,
  $(\Box 0)(y)$ is the degree to which a state~$y$ in a model~$\CA$ is
  deadlocked, i.e.\ the infimum over $1-r$ where $r$ ranges over the
  degrees $R^\CA(y,z)$ to which any state~$z$ is a successor
  of~$y$. Then, $(\Diamond\Box 0)(x)$ is the supremum of
  $\min(R^\CA(x,y),(\Box 0)(y))$ over all~$y$.
  % Details: ([]0)(y) = inf (\neg R(y,z) or 0) = inf neg R(y,z). 

  In the reading of fuzzy~$K$ as the description logic
  fuzzy~$\ALC$~\cite{Straccia98} (with only one role for simplicity),
  the underlying fuzzy relation would be seen as a vague connection
  between individuals, such as a `likes' relation between persons. In
  this reading, the formula
  \begin{equation*}
    \Box(\textsf{soft-spoken}\land\Diamond\textsf{reasonable})
  \end{equation*}
  describes people who only like people who are soft-spoken and like
  some reasonable person, with all these terms understood in a vague
  sense.
\end{expl}
\noindent As indicated previously, the first-order correspondence
language for fuzzy modal logic in this sense is \emph{Zadeh fuzzy
  first-order logic} over a single binary predicate~$R$ and a unary
predicate~$p$ for every propositional atom~$p$. Formulas $\phi,\psi$
of what we briefly term \emph{fuzzy first-order logic} or \emph{fuzzy
  FOL} are thus given by the grammar
\begin{equation*}
  \phi,\psi :: = c\mid p(x)\mid R(x,y)\mid x=y\mid\phi\ominus c\mid\neg\phi\mid\phi\land\psi\mid\exists x.\phi
\end{equation*}
where $c\in[0,1]\cap\Rat$, $p\in\At$, and $x,y$ range over a fixed
countably infinite set of variables. We have the usual notions of free
and bound variables. The \emph{quantifier rank} $\qr(\phi)$ of a
formula~$\phi$ is defined, as usual, as the maximal nesting depth of
quantifiers in~$\phi$ (unlike for the modal rank, we do not let atomic
formulas count towards the quantifier rank). The semantics is
determined as the evident extension of the modal semantics, with the
existential quantifier interpreted as supremum and `$=$' as crisp
equality. Formally, a formula $\phi(x_1,\dots,x_n)$ with free
variables among $x_1,\dots,x_n$ is interpreted, given a fuzzy
relational model~$\CA$ and a vector $\bar a=(a_1,\dots, a_n)$ of
values for the free variables, as a truth value
$\phi(\bar a)\in[0,1]$, given by
\begin{gather*}
  p(x_i)(\bar a)  = p^\CA(a_i) \qquad R(x_i,x_j)(\bar a)=R^\CA(a_i,a_j)\\
  (x_i = x_j)(\bar a)  = 1  \text{ if $a_i=a_j$, and $0$ otherwise}\\
  (\exists x_0.\phi(x_0,\dots,x_n))(\bar a) = \textstyle
  \bigvee_{a_0\in A}\phi(a_0,\bar a)
\end{gather*}
and essentially the same clauses as in the modal case for the other
connectives.

We thus have a variant of the classical \emph{standard translation},
that is, a truth-value preserving embedding $\ST_x$ of fuzzy~$K$ into
fuzzy FOL, indexed over a variable~$x$ naming the current state and
defined inductively by $\ST_x(p)=p(x)$,
\begin{equation*}
  \ST_x(\Diamond\phi)=\exists y.\,(R(x,y)\land\ST_y(\phi)),
\end{equation*}
and commutation with all other constructs. Fuzzy~$K$ thus becomes a
fragment of fuzzy FOL, and the object of the present paper is to
characterize their relationship.

\paragraph{Coalgebraic view} Recall that an \emph{$F$-coalgebra}
$(A,\alpha)$ for a set functor $F:\Set\to\Set$ consists of a set~$A$
of \emph{states} and a map $\alpha:A\to FA$. The set $FA$ is thought
of as containing structured collections over~$A$, so that $\alpha$
assigns to each state~$a$ a structured collection~$\alpha(a)$ of
successors.  Coalgebras thus provide a general framework for
state-based systems~\cite{Rutten00}.  We will partly use coalgebraic
techniques in our proofs, in particular final chain arguments. We
therefore note that fuzzy relational models are coalgebras for the set
functor $\mfun$ given by
\begin{equation*}
  \mfun= [0,1]^\At\times
  \ffun 
\end{equation*}
where $\ffun X=[0,1]^X$ is the fuzzy version of the covariant powerset
functor. That is, $\ffun$ acts on maps $f: X \to Y$ by taking
fuzzy direct images,
\begin{equation*}\textstyle
  \ffun f(g)(y)=\bigvee_{f(x)=y}g(x).
\end{equation*}
Explicitly, a fuzzy relational model $\CA$ corresponds to the
$\mfun$-coalgebra $(A,\alpha)$ given by $\alpha(a)=(h,g)$ where
$h(p)=p^\CA(a)$ for $p\in\At$ and $g(a')=R^\CA(a,a')$ for $a'\in A$.

A \emph{$\mfun$-coalgebra morphism} $f:(A,\alpha)\to(B,\beta)$ between
$\mfun$-coalgebras $(A,\alpha),(B,\beta)$ (i.e.\ fuzzy relational
models $\CA,\CB$) is a map $f:A\to B$ such that
$\mfun f\alpha=\beta f$. Explicitly, this means that~$f$ is a
\emph{bounded morphism}, i.e.\ $p^\CB(f(a))=p^\CA(a)$ for all atoms
$p$ and all $a\in A$, and $R^\CB(f(a),b)=\bigvee_{f(a')=b}R^\CA(a,a')$
for every $b\in B$. For models $\CA$ and $\CB$, we define their
\emph{disjoint union} $\CA+\CB$ as the model with domain $A+B$
(disjoint union of sets), $p^{\CA+\CB}(c) = p^\CA(c)$ for $c\in A$ and
$p^{\CA+\CB}(c)=p^\CB(c)$ otherwise, and
$R^{\CA+\CB}(c,c') = R^\CA(c,c')$ if $c,c'\in A$,
$R^{\CA+\CB}(c,c') = R^\CB(c,c')$ if $c,c'\in B$,
$R^{\CA+\CB}(c,c') = 0$ otherwise. This is precisely the categorical
coproduct of $\CA$ and $\CB$ as $\mfun$-coalgebras; in particular, the
injection maps $\CA\to\CA+\CB$ and $\CB\to\CA+\CB$ are bounded
morphisms.

\section{Pseudometric Spaces}
\label{sec:pseudometric-spaces}

\noindent We recall some basics on pseudometric spaces, which differ
from metric spaces in that distinct points can have distance~$0$:
\begin{definition}[Pseudometric space, non-expansive maps]
  Given a non-empty set $X$, a \emph{(bounded) pseudometric on $X$} is
  a function $d\colon X\times X\to [0,1]$ such that for all
  $x,y,z\in X$, the following axioms hold: $d(x,x) = 0$
  (\emph{reflexivity}), $d(x,y) = d(y,x)$ (\emph{symmetry}),
  $d(x,z) \le d(x,y)+d(y,z)$ (\emph{triangle inequality}). If
  additionally $d(x,y)=0$ implies $x=y$, then $d$ is  a
  \emph{metric}. A \emph{(pseudo)metric space} is a pair $(X,d)$ where
  $X$ is a set and $d$ is a (pseudo)metric on $X$. The \emph{diameter}
  of $A\subseteq X$ is $\bigvee_{x,y\in A}d(x,y)$. We equip the unit
  interval $[0,1]$ with the standard Euclidean distance~$d_e$,
  \begin{equation*}
    d_e(x,y)=|x-y|.
  \end{equation*}
  A function $f\colon X\to Y$ between pseudometric spaces $(X,d_1)$,
  $(Y,d_2)$ is \emph{non-expansive} if
  $d_2 \circ (f \times f) \le d_1$, i.e.\ $d_2(f(x),f(y))\le d_1(x,y)$
  for all $x,y$. We then write
  \begin{equation*}
    f \colon \nonexp{(X, d_1)}{(Y, d_2)}.
  \end{equation*}
  The space of non-expansive functions $\nonexp{(X, d_1)}{(Y, d_2)}$
  is equipped with the \emph{supremum (pseudo)metric} $d_\infty$
  defined by
  \begin{equation*}
    d_\infty(f,g) = \sup_{x\in X} d_2(f(x),g(x))
  \end{equation*}
  In the special case $(Y,d_2) = ([0,1],d_e)$, we will also denote
  $d_\infty(f,g)$ as $\supnorm{f-g}$.
  
  As usual we denote by
  $\ball{d}{\epsilon}{a} = \{x\in X\mid d(a,x) \le \epsilon\}$ the
  \emph{ball} of radius $\epsilon$ around $a$ in $(X,d)$. The space
  $(X,d)$ is \emph{totally bounded} if for every $\epsilon > 0$ there
  exists a finite \emph{$\epsilon$-cover}, i.e.\ finitely many
  elements $a_1,\dots,a_n\in X$ such that
  $X = \bigcup_{i=1}^n \ball{d}{\epsilon}{a_i}$.
\end{definition}
\noindent Recall that a metric space is compact iff it is complete and
totally bounded.

Given a fuzzy relational model $\CA$, we extend the semantics of
$\Diamond$ to arbitrary functions $f\colon A\to[0,1]$ by
\begin{equation*}
  \Diamond f\colon A\to[0,1], 
  \quad
  (\Diamond f)(a)= \bigvee_{a'\in A} R^\CA(a,a') \land f(a').
\end{equation*}

\begin{lem}
  \label{lem:diamond-nonexp}
  The map $f\mapsto\Diamond f$ is non-expansive.
\end{lem}

\section{Behavioural Distance\\and Bisimulation Games}
  \label{sec:bisimulation-games}
  We proceed to define our notion of behavioural distance for fuzzy
  relational models. We opt for a game-based definition as the basic
  notion, and relate it to logical distance, showing that fuzzy modal
  logic is non-expansive w.r.t.\ behavioural distance. In
  Section~\ref{sec:modal-approx} we will give an equivalent
  characterization in terms of fixed points, and show that all
  distances coincide at finite depth. Following ideas used in
  probabilistic bisimulation metrics~\cite{DesharnaisEA08}, we use
  bisimulation games that have crisp outcomes but are parametrized
  over a maximal allowed deviation; we will then define the distance
  as the least parameter for which duplicator wins.
\begin{defn}
  \label{def:bisimulation-game}
  Let $\CA,\CB$ be fuzzy relational models, and let
  $a_0 \in A, b_0 \in B$.  The \emph{$\epsilon$-bisimulation game for
    $\CA,a_0$ and $\CB,b_0$} (or just for $a_0,b_0$) played by $S$
  (\emph{spoiler}) and $D$ (\emph{duplicator}) is given as follows.
  \begin{itemize}
  \item \emph{Configurations:} pairs $(a,b)\in A\times B$ of states.
  \item \emph{Initial configuration:} $(a_0,b_0)$.
  \item \emph{Moves:} Player $S$ needs to pick a new state in one of
    the models $\CA$ or $\CB$, say $a'\in A$, such that
    $R^\CA(a,a') > \epsilon$, and then $D$ needs to pick a state in
    the other model, say $b'\in B$, such that
    $R^\CB(b,b')\ge R^\CA(a,a')-\epsilon$. The new configuration is
    then $(a',b')$.
  \item \emph{Winning condition:} Any player who needs to move but
    cannot, loses. Player $D$ additionally needs to maintain the
    following winning condition \emph{before} every round: For every
    $p\in\At$, $|p^\CA(a)-p^\CB(b)|\le\epsilon$.
  \end{itemize}
  There are two variants of the game, the unrestricted game in which
  $D$ wins infinite plays, and the \emph{depth-$n$
    $\epsilon$-bisimulation game}, which is restricted to~$n$ rounds,
  meaning that $D$ wins after $n$ rounds have been played.
\end{defn}

\begin{rem}
  \label{rem:zero-game}
  Note that, since the invariant only needs to hold before every round
  actually played,~$D$ always wins the depth-$0$ game
  regardless of $a_0$ and $b_0$.
\end{rem}
\noindent 
The usual composition lemma for bisimulations then takes the following
form:
\begin{lem}
  \label{lem:game-transitive}
  Let $\CA,\CB,\mathcal{C}$ be models and
  $a_0\in A, b_0\in B, c_0\in C$ such that $D$ wins the
  $\epsilon$-bisimulation game for $(a_0,b_0)$ and the
  $\delta$-bisimulation game for $(b_0,c_0)$. Then $D$ also wins the
  $(\epsilon+\delta)$-bisimulation game for $(a_0,c_0)$. The same
  holds for the corresponding depth-$n$ bisimulation games.
\end{lem}
% \begin{proof}[Proof (sketch)]
%   Player $D$ wins by following the two existing winning strategies in
%   parallel.
% \end{proof}
\noindent As indicated above, we then obtain a notion of behavioural
distance by taking infima:

\begin{defn}[Behavioural distance]
  \label{def:bisim-distance}
  Let $\CA,a_0$ and $\CB,b_0$ be as in
  Definition~\ref{def:bisimulation-game}.  The \emph{behavioural
    distance} $d^G(a_0,b_0)$ of $a_0$ and $b_0$ is the infimum over
  all $\epsilon$ such that $D$ wins the $\epsilon$-bisimulation game
  for $a_0$ and $b_0$.  The \emph{depth-$n$ behavioural distance}
  $d^G_n(a_0,b_0)$ of $a_0$ and $b_0$ is defined analogously, using
  the depth-$n$ bisimulation game. 
\end{defn}
\noindent This definition is justified by the following lemma, which
follows from Lemma~\ref{lem:game-transitive}:
\begin{lem}\label{lem:pseudometrics}
  The behavioural distance $d^G$ and all depth-$n$ behavioural
  distances $d^G_n$ are pseudometrics.
\end{lem}

\begin{rem}
  We emphasize that $d^G(a,b)=0$ does not in general imply that~$D$
  wins the $0$-bisimulation game on $a,b$. In this sense, the notion
  of $\epsilon$-bisimulation is thus what enables us to avoid
  restricting to models that are \emph{witnessed}~\cite{Hajek05} in
  the sense that all suprema appearing in the evaluation of
  existential quantifiers are actually maxima.
\end{rem}
\noindent We have the expected relationship between the various
behavioural pseudometrics:
\begin{lem}\label{lem:metrics}
  For all models $\CA,\CB$, states $a\in A$, $b\in B$, and
  $n\ge m\ge 0$, we have
  \begin{equation*}
    d^G_m(a,b)\le d^G_n(a,b)\le d^G(a,b).
  \end{equation*}
\end{lem}
\noindent As usual, behavioural equivalence is invariant under
coalgebra morphisms; this can now be phrased as follows:
\begin{lem}\label{lem:mor-zero-dist}
  Let $\CA,\CB$ be fuzzy relational models, and let $f:\CA\to\CB$ be a
  bounded morphism. Then for every $a\in A$, $d^G(a,f(a))=0$.
\end{lem}
\begin{proof}[Proof (sketch)]
  Player $D$ wins the depth-$n$ $\epsilon$-bisimulation game for every
  $\epsilon>0$.
\end{proof}
% \noindent By the triangle inequality, the following is then immediate:
% \begin{lem}
%   Bounded morphisms are isometries w.r.t.\ behavioural distance.
% \end{lem}
\noindent Since coproduct injections are bounded morphisms, a special
case is
\begin{lem}
  \label{lem:bisim-inv-disjoint-union}
  Given models $\CA,\CB$ and $a\in A$, the state $a$ in~$\CA$ and the
  corresponding state $a$ in $\CA+\CB$ have behavioural distance~$0$.
\end{lem}

\noindent Behavioural distance determines our notion of bisimulation
invariance, which we take to mean non-expansivity w.r.t.\ behavioural
distance. To match this with the standard notion, interpret classical
crisp bisimilarity as a discrete pseudometric~$d$ assigning distance
$0$ to pairs of bisimilar states and $1$ to non-bisimilar ones, and
similarly interpret crisp predicates~$P$ as maps into $\{0,1\}$; then
$P$ is bisimulation-invariant in the usual sense iff~$P$ is
non-expansive w.r.t.\ $d$. Formal definitions for the fuzzy setting
are as follows.
\begin{defn}[Bisimulation-invariant formulas and predicates]
  A formula $\phi$ (either in fuzzy modal logic or in fuzzy FOL, with
  a single free variable) is \emph{bisimulation-invariant} if for all
  models $\CA,\CB$ and all states $a\in A$, $b\in B$,
  \[ |\phi(a) - \phi(b)| \le d^G(a,b). \]
  Similarly, given a model $\CA$, a \emph{(fuzzy) state predicate}
  on~$A$, i.e.\ a function $P\colon A\to [0,1]$, is
  \emph{bisimulation-invariant} if~$P$ is non-expansive w.r.t.\ the
  bisimulation distance $d^G$. In both cases, \emph{depth-$n$
    bisimulation invariance} is defined in the same way using
  depth-$n$ behavioural distance.
\end{defn}
\noindent As expected, Zadeh fuzzy modal logic is bisimulation
invariant; more precisely:
\begin{lem}[Bisimulation invariance]
  \label{lem:bisim-inv-modal}
  Every fuzzy modal formula of rank at most $n$ is depth-$n$
  bisimulation-invariant. 
\end{lem}

% \begin{proof}[Proof (sketch)]
%   Suppose that $D$ wins the $\epsilon$-bisimulation-game for $(a,b)$.
%   We show that $|\phi(a)-\phi(b)|\le\epsilon$ for all
%   $\phi\in\modf{n}$ by induction on $\phi$. The case for propositional
%   atoms~$p$ follows from the winning condition, and the propositional
%   operators are all themselves non-expansive.  We treat the case for
%   the modality~$\Diamond$: By symmetry, it suffices to show that
%   $(\Diamond\phi)(b) \ge (\Diamond\phi)(a)-\epsilon$. We can assume
%   $(\Diamond\phi)(a)>\epsilon$.  Let $\delta>0$; then there is~$a'$
%   such that $R(a,a')>\epsilon$ and
%   $(\Diamond\phi)(a)-(R(a,a')\land\phi(a'))<\delta$.  Let $b'$ be
%   $D$'s winning answer to $S$'s move $a'$.  Then by induction,
%   $|\phi(a')-\phi(b')|\le\epsilon$, and moreover
%   $R(b,b')\ge R(a,a')-\epsilon$ by the rules of the game. Thus,
%   \begin{multline*}
%     (\Diamond\phi)(b)\ge
%     R(b,b')\land\phi(b')\ge
%     (R(a,a')-\epsilon)\land(\phi(a')-\epsilon)\\
%     =(R(a,a')\land\phi(a'))-\epsilon
%     >(\Diamond\phi)(a)-\epsilon-\delta.
%   \end{multline*}
%   Since this holds for every $\delta$, it follows that
%   $(\Diamond\phi))(b)\ge(\Diamond\phi))(a)-\epsilon$.
%   %We show symmetrically that
%   %$(\Diamond\phi)(a)\ge(\Diamond\phi)(b)-\epsilon$, that is,
%   %$|(\Diamond\phi)(a)-(\Diamond\phi)(b)|\le\epsilon$ as required.
% \end{proof}
\noindent In particular, for every rank-$n$ modal formula $\phi$ and
every fuzzy relational model $\CA$, the evaluation map
$\phi_\CA:A\to[0,1]$ is a non-expansive map
$\nonexp{(A,d^G_n)}{([0,1],d_e)}$. A forteriori
(Lemma~\ref{lem:metrics}), \emph{every fuzzy modal formula~$\phi$ is
  bisimulation-invariant}, i.e.\ $\phi(\cdot)$ is non-expansive
w.r.t.\ (unbounded-depth) behavioural distance $d^G$. 
\begin{expl}
  The formula $R(x,x)$ in fuzzy FOL fails to be bisimulation-invariant
  (compare a loop with an infinite chain), and is therefore neither
  expressible nor approximable by fuzzy modal formulas.
\end{expl}

\begin{defn}[Logical distance] 
  We further define \emph{logical distances} $d^L$ (w.r.t.\ all modal
  formulas) and $d^L_n$ (w.r.t.\ modal formulas of rank at most~$n$)
  by
  \begin{align*}
    d^L(a,b)&=\textstyle\bigvee_{\phi\text{ modal}}|\phi(a)-\phi(b)|,\\ 
    d^L_n(a,b) & = \textstyle\bigvee_{\rk(\phi) \le n} |\phi(a) - \phi(b)|.
  \end{align*}
\end{defn}
\noindent We clearly have
\begin{equation*}
  d^L_m(a,b)\le d^L_n(a,b)\le d^L(a,b)\quad\text{for $n\ge m\ge 0$},
\end{equation*}
as well as
\begin{equation}\label{eq:log-dist-sup}
  d^L(a,b)=\textstyle\bigvee_{n\ge 0}d^L_n(a,b).
\end{equation}
Using~\eqref{eq:log-dist-sup} and Lemma~\ref{lem:metrics}, we can then
rephrase bisimulation invariance (Lemma~\ref{lem:bisim-inv-modal}) as
\begin{lem}\label{lem:bisim-inv-dist}
  For models $\CA,\CB$, states $a\in A$, $b\in B$, and $n\ge 0$, we
  have
  \begin{equation*}
    d^L_n(a,b)\le d^G_n(a,b)\quad\text{and}\quad
    d^L(a,b)\le d^G(a,b).
  \end{equation*}
\end{lem}

\begin{rem}
  Under \L{}ukasiewicz semantics (Remark~\ref{rem:zadeh}),
  non-expansivity clearly breaks; e.g.\ if $a$ and $b$ are states
  without successors in models $\CA$ and $\CB$, respectively, such
  that $p^\CA(a)=0.9$, $p^\CA(b)=0.8$, and $a$ and $b$ agree on all
  other atoms, then $d^G(a,b)=0.1$ but $|\phi(a)-\phi(b)|=0.2$ for the
  formula $\phi=p\land p$, since under \L{}ukasiewicz semantics,
  $\phi(a)=p^\CA(a)+p^\CA(a)-1 = 0.8$ and
  $\phi(b)=p^\CB(b)+p^\CB(b)-1 = 0.6$. See also a similar example
  in~\cite{fl:regular-equ-social-networks}. For a treatment of
  \L{}ukasiewicz fuzzy modal logic, one would thus need to replace
  non-expansivity with Lipschitz continuity (see
  also~\cite{SchroderPattinson11}). Additional problems, however,
  arise with logical distance: Defining a logical distance for
  \L{}ukasiewicz modal logic in analogy to the above definition of
  $d^L$ gives a discrete pseudometric. The reason is that small behavioural
  differences between models can be amplified arbitrarily in
  \L{}ukasiewicz logic using conjunction, as illustrated precisely by
  the above example (where we could also use $p\land p\land p$
  etc.). The statement of a van Benthem theorem for
  \L{}ukasiewicz modal logic would thus presumably become quite
  complicated, e.g.\ would need to stratify over Lipschitz constants.
\end{rem}
\noindent We now launch into the proof of our target result, which
states that \emph{every bisimulation-invariant fuzzy first-order
  property can be approximated by fuzzy modal formulas}, a converse to
bisimulation-invariance of fuzzy modal formulas. As already indicated,
we follow a proof strategy established for the classical setting by
Otto~\cite{o:van-Benthem-Rosen-elementary}: We show that
\begin{itemize}
\item every bisimulation-invariant fuzzy first-order property is
  \emph{$\ell$-local} for some $\ell$, w.r.t.\ a suitable notion of
  Gaifman distance (Section~\ref{sec:locality});
\item every $\ell$-local bisimulation-invariant fuzzy first-order
  property is already depth-$n$ bisimulation-invariant for some~$n$
  (Section~\ref{sec:main}); and
\item every depth-$n$ bisimulation-invariant fuzzy state property is
  approximable by fuzzy modal formulas of rank at most~$n$
  (Section~\ref{sec:modal-approx}).
\end{itemize}
We begin with the last step of this program.

\section{Modal Approximation at Finite Depth}\label{sec:modal-approx}

Having seen game-based and logical behavioural distances $d^L$, $d^G$
in the previous section, we proceed to introduce a third, fixed-point
based definition, and then show that all three distances agree at
finite depth. This happens in a large simultaneous induction, in which
we also prove that \emph{every depth-$n$ bisimulation-invariant fuzzy
  state property is approximable by modal formulas of rank~$n$}. As
indicated in the introduction, this is in fact the technical core of
the paper. This is in sharp contrast with the classical setting, where
the corresponding statement -- every depth-$n$ bisimulation-invariant
crisp state property is expressible by a crisp modal formula -- is
completely straightforward.
\begin{assn}\label{ass:fin-atoms}
  As usual in proofs of van Benthem type results, we assume from now
  on that $\At=\{p_1,\dots,p_k\}$ is finite. This is w.l.o.g.\ for
  purposes of the proof of our main result, as we will aim to show
  modal approximability of a given formula, so only finitely many
  atoms are relevant. Note that, e.g.,
  Theorem~\ref{thm:modal-approx}.\ref{item:tot-bounded} (total boundedness
  of finite-depth behavioural distance) will presumably not hold
  without this assumption.
\end{assn}
\noindent The fixed-point definition of behavioural distance is based
on the \emph{Kantorovich
  lifting}~\cite{bbkk:behavioral-metrics-functor-lifting}.  We first
define an \emph{evaluation function}
\begin{math}
\mathit{ev}\colon \ffun[0,1]\to [0,1]
\end{math}
(recall from Section~\ref{sec:logics} that $\ffun X=[0,1]^X$ and
$\mfun X=[0,1]^\At\times\ffun X$) by
\begin{equation*}
\mathit{ev}(g) = \textstyle\bigvee_{s\in[0,1]} (g(s)\land s)\quad\text{for
$g\colon [0,1]\to [0,1]$}.
\end{equation*}
Given a pseudometric space $(X,d)$, we define the \emph{Kantorovich
  pseudometrics} $d^\ffun$ on $\ffun X$ and $d^\mfun$ on $\mfun X$,
respectively, by
\begin{gather*}
   d^\ffun(g_1,g_2) =
   \quad\bigvee_{\mathclap{f: \nonexp{(X,d)}{([0,1],d_e)}}}\quad
   |\mathit{ev}\circ\ffun f(g_1) - \mathit{ev}\circ\ffun f(g_2)| \\
   d^\mfun((r_1,g_1),(r_2,g_2)) =
   d^\ffun(g_1,g_2)\lor\textstyle\bigvee_{p\in\At}|r_1(p)-r_2(p)| 
\end{gather*}
for $r_i \in [0,1]^\At$ and $g_i \in \ffun X$. It follows from general
results on lifting metrics along
functors~\cite{bbkk:behavioral-metrics-functor-lifting} that $d^\ffun$
and $d^\mfun$ are indeed pseudometrics. % , and moreover extend to
% functors $\bar\ffun,\bar\mfun$ on pseudometric spaces, i.e.\ whenever
% $f:(X,d_1)\to (Y,d_2)$ is non-expansive, then so are
% $\bar\ffun f:=\ffun f: (\ffun X,d_1^\ffun)\to(\ffun Y,d_2^\ffun)$ and
% $\bar\mfun f:=\mfun f: (\mfun X,d_1^\mfun)\to(\mfun Y,d_2^\mfun)$.
% These functors are called the \emph{Kantorovich liftings} of $\ffun$
% and $\mfun$, respectively.
% maps and then yield , where the two-step approach used above
% (first lift $\ffun$
% then $\mfun$) is taken care of by multifunctor lifting.

Given a fuzzy relational model $\CA = (A,(p^\CA)_{p\in\At},R^\CA)$,
viewed as a coalgebra $\alpha\colon A\to \mfun A$ as discussed in
Section~\ref{sec:logics}, we can inductively define a sequence of
pseudometrics $(d^K_n)_{n\ge 0}$ on $A$ via the Kantorovich lifting:
\begin{equation*}
  d^K_0(a,b) = 0 
  \text{\quad and \quad} 
  d^K_{n+1} = (d^K_n)^\mfun \circ (\alpha\times\alpha),
\end{equation*}
that is, expanding definitions,
\begin{equation*}
  d^K_{n+1}(a,b) = \bigvee_{p \in \textrm{At}}|p(a)-p(b)| \lor
  \bigvee_{\mathclap{f: \nonexp{(A,d^K_n)}{([0,1],d_e)}}}
  | (\Diamond f)(a) - (\Diamond f)(b) | .
\end{equation*}
(The $d^K_n$ can be seen as approximants of a fixed point, which is
not itself needed here.)

%where $(\Diamond f)(a) = \bigvee_{a'\in A} R(a,a') \land f(a')$.

\begin{expl}
  To illustrate the three forms of behavioural distance (logical,
  game-based, and via the Kantorovich lifting), we use the following
  model $\mathcal{A}$ with one propositional atom $p$. In the diagram
  below, each state has the form $x[p:p^\mathcal{A}(x)]$, and
  transitions from $x$ to $y$ are labelled with their truth values
  $R^\mathcal{A}(x,y)$.
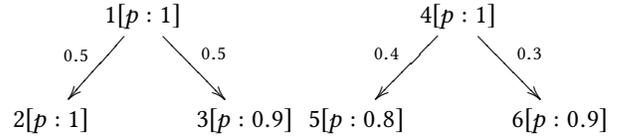
  \[
    \xymatrix@C=0mm{
      & 1[p:1] \ar[dl]_{0.5} \ar[dr]^{0.5} & & & 4[p:1] \ar[dl]_{0.4}
      \ar[dr]^{0.3} & \\
      2[p:1] & & 3[p:0.9] & 5[p:0.8] & & 6[p:0.9]
    }
  \]
  Clearly, it suffices to look at depth~$2$ in this example. The
  game-based distance of $1,4$ is $d^G(1,4)=d_2^G(1,4) = 0.2$. To see this,
  first note that~$D$ has a winning strategy for $\epsilon = 0.2$:
  Player~$S$ may pick any transition, and $D$ then always has a
  transition available as a reply; irrespective of their choices, they
  end up in a pair of states with values of $p$ differing by at most
  $0.2$, and then~$S$ needs to move but cannot. The situation is
  different for $\epsilon < 0.2$: In this case,~$S$ can take the
  transition from $1$ to $2$, which~$D$ must answer by going from $4$
  to $5$, since $R^\CA(4,6) \ngeq R^\CA(1,2) - \epsilon$. But
  $|p^\CA(2)-p^\CA(5)|=0.2>\epsilon$, so $S$ wins.

  This distance is witnessed by the formula
  $\phi = \Diamond (p\ominus 0.5)$:
  \begin{align*}
    \phi(1) & =  (0.5 \land (p^\CA(2)\ominus 0.5))
    \lor (0.5 \land (p^\CA(3)\ominus 0.5)) = 0.5 \\
    \phi(4) & =  (0.4 \land (p^\CA(5)\ominus 0.5))
    \lor (0.3 \land (p^\CA(6)\ominus 0.5)) = 0.3,
  \end{align*}
  so $d^L(1,4)=d_2^L(1,4)=0.2$ (recall $\rk(\phi)=2$) by
  Lemma~\ref{lem:bisim-inv-dist}. Note that $\Diamond p$ would only
  yield a difference of $0.1$.

  As to the Kantorovich distances, we have $d^K_0(1,4)=0$, so that
  the~$f$ over which the supremum in the definition of $d^K_1(1,4)$ is
  taken are all constant; it is then easily seen that
  $d^K_1(1,4)=0.1$, the difference between the maximal transition
  degree from~$1$ ($0.5$) and that from~$4$ ($0.4$). The function
  corresponding to $p\ominus 0.5$ then serves as a witness of the
  behavioural distance at depth~$2$, so that $d_2^K(1,4) \ge 0.2$; one
  can check that in fact $d_2^K(1,4)= 0.2$.
\end{expl}
\noindent The main result proved in this section is then the following
theorem, which as indicated above states in particular that the
definitions of behavioural distance coincide at finite depth and that
the modal formulas lie dense in the non-expansive state properties:
\begin{thm}
  \label{thm:modal-approx}
  Let $\CA$ be a fuzzy relational model.  Then the following holds for
  all $n \ge 0$.
  \begin{enumerate}
    \item\label{item:metrics-equal}
      $d^G_n = d^L_n = d^K_n =: d_n$ on $\CA$.
    \item\label{item:tot-bounded}
      The pseudometric space $(A,d_n)$ is totally bounded.
    \item\label{item:modal-approx}
      The modal formulas of rank at most $n$ form a dense subset of
      the space $\nonexp{(A,d_n)}{([0,1],d_e)}$.
  \end{enumerate}
\end{thm}
\noindent (We note that the equality $d^L_n=d^G_n$ is effectively the
finite-depth part of a Hennessy-Milner property; the infinite-depth
version will, of course, hold only under finite branching. This
contrasts somewhat with the probabilistic
case~\cite{BreugelWorrell05}.)

\begin{proof}[Proof (sketch)]%\renewcommand{\qedsymbol}{}
  We prove all claims simultaneously by induction on $n$. The base
  case $n=0$ is trivial. The proof of the induction step is split over
  a number of lemmas proved next:
  \begin{itemize}
    \item Item~\ref{item:metrics-equal} is proved in
    Lemmas~\ref{lem:metrics-equal-LK} and \ref{lem:metrics-equal-GK}.
    \item Item~\ref{item:tot-bounded} is proved in
    Lemma~\ref{lem:tot-bounded}.
  \item Item~\ref{item:modal-approx} is proved in
    Lemma~\ref{lem:modal-approx}.\qedhere
  \end{itemize}
\end{proof}
\noindent For the remainder of this section, we fix a model $\CA$ as
in Theorem~\ref{thm:modal-approx} and $n > 0$, and assume as inductive
hypothesis that all claims in Theorem~\ref{thm:modal-approx} already
hold for all $n' < n$.

\begin{lem}\label{lem:metrics-equal-LK}
  We have $d^L_n = d^K_n$ on $\CA$.
\end{lem}
\begin{proof}[Proof (sketch)]
  Let $a,b\in A$ and put $F:=\nonexp{(A,d_{n-1})}{([0,1],d_e)}$. By
  Lemma~\ref{lem:diamond-nonexp}, the map
  \[ H\colon (F,d_\infty) \to ([0,1],d_e), f \mapsto |(\Diamond f)(a)
  - (\Diamond f)(b)| \]
  is continuous. Since by the induction hypothesis, $\modf{n-1}$ is
  dense in $F$, it follows that $H[\modf{n-1}]$ is dense in
  $H[F]$. Thus,
  \begin{align*}
    d^K_n(a,b) = & \bigvee_{p \in \textrm{At}}|p(a)-p(b)| \lor
                   \bigvee_{\mathclap{f: \nonexp{(A,d_{n-1})}{([0,1],d_e)}}}
                   | (\Diamond f)(a) - (\Diamond f)(b) | \\
    = & \bigvee_{p \in \textrm{At}}|p(a)-p(b)| \lor
  \bigvee_{\mathclap{\rk\phi\le n-1}}
        | (\Diamond \phi)(a) - (\Diamond \phi)(b) | \\
    = & \bigvee_{\mathclap{\rk\phi\le n}}|\phi(a) - \phi(b)|
        = d^L_n(a,b). \qedhere
  \end{align*}
\end{proof}

\begin{lem}\label{lem:metrics-equal-GK}
  We have $d^G_n = d^K_n$ on $\CA$.
\end{lem}
\begin{proof}[Proof (sketch)]
  First, $d^K_n(a,b) = d^L_n(a,b) \le d^G_n(a,b)$ for all $a,b$ by
  Lemmas~\ref{lem:metrics-equal-LK} and~\ref{lem:bisim-inv-modal}.
  
  In the other direction, if $d^K_n(a,b)\le\epsilon$, we need to show
  that~$D$ wins the $(\epsilon+\delta)$-game on $(a,b)$ for all
  $\delta>0$. W.l.o.g.~$S$ moves from $a$ to some $a'$. We can
  instantiate the function $f$ in the definition of $d^K_n$ as
  \begin{equation*}
    f\colon \nonexp{(A,d_{n-1})}{([0,1],d_e)},
    f(b') = R(a,a') \ominus d_{n-1}(a',b').
  \end{equation*}
  A winning reply for $D$ can now be extracted by taking a state $b'$
  that approximates the supremum in $(\Diamond f)(b)$ sufficiently
  closely.  One checks that $b'$ is a legal move and that
  $d_{n-1}(a',b')<\epsilon+\delta$, so~$D$ wins.
\end{proof}
\noindent Having shown that the pseudometrics $d_n^L$, $d_n^G$,
$d_n^K$ coincide, we will now use $d_n$ to denote any of them, as
indicated in Theorem~\ref{thm:modal-approx}.

The next lemma is a version of the Arzel\`a-Ascoli theorem for total
boundedness instead of compactness and non-expansive instead of
continuous functions; that is, we impose weaker assumptions on the
space but stronger assumptions on the functions.

\begin{lem}\label{lem:arzela-ascoli}
  Let $(X,d_1),(Y,d_2)$ be totally bounded pseudometric spaces. Then
  the space $\nonexp{(X,d_1)}{(Y,d_2)}$, equipped with the supremum
  pseudometric, is totally bounded.
%\[ d_\infty(f,g) = \bigvee_{x\in X}d_2(f(x),g(x)). \]
\end{lem}

\noindent The following lemma, the inductive step for
Theorem~\ref{thm:modal-approx}.\ref{item:tot-bounded}, then guarantees
that our variant of Arzel\`a-Ascoli will actually apply to
$(A,d_n)$ in the next round of the induction.
\begin{lem}\label{lem:tot-bounded}
  $(A,d_n)$ is a totally bounded pseudometric space.
  %The pseudometric space $(A,d_n)$ is totally bounded.
\end{lem}

\begin{proof}[Proof (sketch)]
Put $F:=\nonexp{(A,d_{n-1})}{([0,1],d_e)}$ and let $\epsilon>0$.
By Lemma~\ref{lem:arzela-ascoli}, $F$ is totally bounded, so as
$\modf{n-1}$ is dense in $F$, there exists a finite
$\frac{\epsilon}{6}$-cover of $F$ consisting of formulas
$\phi_1,\dots,\phi_m \in \modf{n-1}$.
One can now show that the map
\begin{align*}
I\colon A &\to [0,1]^{k+m}\\
  a& \mapsto
(p_1(a),\dots,p_k(a),(\Diamond\phi_1)(a),\dots,(\Diamond\phi_m)(a))
\end{align*}
is an \emph{$\frac{\epsilon}{3}$-isometry}, i.e. for all $a,b \in A$,
\begin{equation*}
|d_n(a,b) - \supnorm{I(a) - I(b)}| \le \tfrac{\epsilon}{3}.
\end{equation*}
Using pre-images under $I$ and a simple triangle inequality argument,
we can then convert a finite $\frac{\epsilon}{3}$-cover of the compact
space $([0,1]^{k+m},d_\infty)$ into a finite $\epsilon$-cover of
$(A,d_n)$.
\end{proof}

\noindent We next prove a variant of the lattice version of the
Stone-Weierstra\ss{} theorem (e.g.~\cite[Lemma
A.7.2]{a:real-analysis-probability}). Again, we only assume the space
to be totally bounded instead of compact but require functions to be
non-expansive rather than only continuous. (A Stone-Weierstra\ss{}
argument appears also in a probabilistic Hennessy-Milner
result~\cite{BreugelWorrell05}).
\begin{lem}
  \label{lem:stone-weierstrass}
  Let $(X,d)$ be a totally bounded pseudometric space, and let $L$ be
  a subset of $F := \nonexp{(X,d)}{([0,1],d_e)}$ such that
  $f_1,f_2 \in L$ implies $\min(f_1,f_2),\max(f_1,f_2) \in L$. If each
  $f\in F$ can be approximated at each pair of points by functions in
  $L$, then $L$ is dense in $F$.
\end{lem}

\begin{lem}\label{lem:modal-approx}
  The modal formulas of rank at most $n$ form a dense subset of
  the space $\nonexp{(A,d_n)}{([0,1],d_e)}$.
\end{lem}

\begin{proof}[Proof (sketch)]
We proceed as in \cite{BreugelWorrell05}, applying
Lemma~\ref{lem:stone-weierstrass} to $\modf{n}$:

Given a function $f\colon\nonexp{(A,d_n)}{([0,1],d_e)}$ and points
$a,b\in A$, a formula $\phi$ approximating $f$ at $a$ and $b$ can be
constructed as follows: Let $\psi\in\modf{n}$ be such that
$|\psi(a)-\psi(b)|$ approximates $|f(a)-f(b)|$ (such a $\psi$ exists
by non-expansivity of $f$). Then $\phi$ is defined from $\psi$ by
means of modified subtraction $\ominus$, which preserves the rank of
formulas.
\end{proof}
\noindent This concludes the proof of
Theorem~\ref{thm:modal-approx}. The theorem still leaves one loose
end: The modal formulas that approximate a given depth-$n$
bisimulation-invariant state property on a model~$\CA$ might depend
on~$\CA$. We eliminate this dependency in the next section, using the
final chain construction.

\section{The Final Chain}\label{sec:final-chain}

\noindent The \emph{final chain}~\cite{Barr93,AdamekKoubek95} of the
functor $\mfun$ is a sequence of sets~$F_k$ that represent all the
possible depth-$k$ behaviours. It is constructed as follows. We take
$F_0$ to be a singleton $F_0=\{*\}$ (reflecting that all states are
equivalent at depth~$0$), and
\begin{equation*}
  F_{n+1} = \mfun F_n = [0,1]^\At\times\ffun F_n.
\end{equation*}
Given a model $\CA$, seen as a coalgebra $\alpha\colon A \to \mfun A$,
we can now define a sequence of projections $\pi_n:A \to F_n$, to be
thought of as mapping states to their depth-$n$ behaviours, by
\begin{equation*}
  \pi_0 = \bang \quad \text{and} \quad \pi_{n+1} = \mfun \pi_n \circ \alpha,
\end{equation*}
where $\bang$ denotes the unique map $A\to F_0$. Explicitly,
$\pi_{n+1}$ is thus defined by
\begin{equation}\label{eq:final-cone}
  \pi_{n+1}(a) = (\lambda p.p^\CA(a),
  \lambda y.\textstyle{\bigvee_{\pi_n(a')=y}R^\CA(a,a')}).
\end{equation}
We next build a model $\CF$ realizing all finite-depth behaviours by
taking the union $F=\bigcup_{k\in\Nat}F_k$ (automatically
disjoint). We define the model structure on $F$ by letting every
element behave as it claims to: For
$(h,g)\in F_{k+1}=[0,1]^\At\times\ffun F_k$ and $y\in F$, we put
$p^\CF(h,g)=h(p)$ for $p\in\At$ and
\begin{equation*}
  R^\CF((h,g),y)=
  g(y) \quad\text{if $y \in F_k$,}
\end{equation*}
and $R^\CF((h,g),y)=0$ otherwise. For $*\in F_0$, we just put
$p^\CF(*)=R^\CF(*,y)=0$.

% \pwnote[inline]{TODO: the following is still very rough, formulate
% this better}

%  We can equip each of them with a coalgebra structure by fixing an
% arbitrary `one step default behaviour' $\cobang\colon 1\to \mfun 1$ and
% applying $\mfun$ appropriately:
% \[ \mfun^n\cobang\colon F_n\to F_{n+1} \]
% We thus get, for every $n\ge 0$, a model $\CF_n$\pwnote{this clashes
% with the notation $\ffun$ for fuzzy powerset} with base set $F_n$.
% We explicitly record the model structure in the non-trivial case: for
% $n\ge 0$ and $z,z'\in F_{n+1}$ with $z=(g,h)$,
% \begin{gather*}
%   p^{\CF_{n+1}}(z) = g(p) \text{ for every } p\in\At \\
%   R^{\CF_{n+1}}(z,z') = 
%     \begin{cases}
%       h(y) & \text{ if } z' = (\mfun^n\cobang)(y) \text{ for some }
%       y\in F_n \\
%       0 & \text{ otherwise}.
%     \end{cases}\\
% \end{gather*}

In the proof of our main result (Theorem~\ref{thm:benthem-rosen}), the
following lemma will allow us to choose approximating modal formulas
uniformly across models.
\begin{lem}
  \label{lem:bisim-inv-projection}
  Let $\CA$ be a model. Then $d^G_n(a,\pi_n(a)) = 0$ for all $a\in A$.
\end{lem}
\begin{proof}[Proof (sketch)]
  Player~$D$ wins the depth-$n$ $\epsilon$-bisimulation game for
  every~$\epsilon>0$ by maintaining the invariant that in round~$i$,
  the configuration has the form $(a',\pi_{n-i}(a'))$ for some
  $a'\in A$. One sees from~\eqref{eq:final-cone} that this invariant
  implies the winning condition and can actually be maintained by~$D$.
\end{proof}
\section{Locality}\label{sec:locality}
\noindent We proceed to show that every bisimulation-invariant formula
of fuzzy FOL is \emph{local}. To this end, we introduce a notion of
Gaifman distance in fuzzy models, as well as a variant of
Ehrenfeucht-Fra\"iss\'e games. The requisite notions of Gaifman graph
and neighbourhood, as well as the definition of locality, are, maybe
unexpectedly, fairly crisp. This is technically owed to the fact that
unlike continuity, non-expansivity does not go well with chains of
$\epsilon$-estimates.

\begin{defn}
  Let $\CA$ be a fuzzy relational model.
  \begin{sparitemize}
  \item The \emph{Gaifman graph of~$\CA$} is an undirected graph with
    set~$A$ of nodes and an edge $\{a,b\}$ for every $a,b \in A$ such
    that $R(a,b) > 0$ or $R(a,b) > 0$.
  \item For every $a,b \in A$, the \emph{Gaifman distance}
    $D(a,b) \in \Nat\cup\{\infty\}$ is the minimal length (i.e.~number
    of edges) of a path between $a$ and $b$ in the Gaifman graph.
  \item For $a \in A$ and $\ell \in \mathbb{N}$, the
    \emph{neighbourhood} of $a$ with radius~$\ell$ is the set
    $\nbhood{\ell}{a}$ given by
    \[ \nbhood{\ell}{a} = \{ b \in A \mid D(a,b) \le \ell \}. \]
    For $\bar a = (a_1,\dots,a_n)$, we put
    $\nbhood{\ell}{\bar a} = \bigcup_{i\le n}\nbhood{\ell}{a_i}$.
  \end{sparitemize}
\end{defn}

\begin{defn}
  Let $\CA$ be a fuzzy relational model and $U \subseteq A$. The
  \emph{restriction} $\CA|_U$ of $\CA$ to $U$ is the fuzzy relational
  model $(U,(p^{\CA|_U})_{p\in\At},R^{\CA|_U})$ with
  $p^{\CA|_U}(a) = p^\CA(a)$ and $R^{\CA|_U}(a,b) = R^\CA(a,b)$ for
  $a,b\in U$. If $U = \nbhood{\ell}{\bar a}$ for some vector $\bar a$
  over~$A$, we also write
  $\rnbhood{\CA}{\ell}{\bar a} := \CA|_\nbhood{\ell}{\bar a}$.
\end{defn}
\noindent As indicated, the ensuing notion of locality is
on-the-nose:

\begin{defn}
  A formula $\phi$ is \emph{$\ell$-local} for $\ell \in \mathbb{N}$ if
  \begin{equation*}
    \phi_\CA(a) = \phi_\rnbhood{\CA}{\ell}{a}(a)
  \end{equation*}
  for every fuzzy relational model $\CA$ and every $a \in A$.
\end{defn}
\noindent It is easy to see that depth-$k$ behaviour depends only on
$k$-neighbourhoods, i.e.
\begin{lem}
  \label{lem:nbhood-bisim}
  For any model $\CA$, $a_0\in A$, and $k > 0$, $D$ wins the depth-$k$
  $0$-bisimulation game for $\CA,a_0$ and $\rnbhood{\CA}{k}{a_0},a_0$.
\end{lem}
\noindent In combination with Lemma~\ref{lem:bisim-inv-modal}, we obtain
\begin{cor}
  \label{lem:modal-local}
  Every fuzzy modal formula of rank at most~$k$ is $k$-local.
\end{cor}
%\begin{proof}
%  We proceed by induction over formulas. The case $\phi = a \in
%  \Rat\cap[0,1]$ is trivial and the cases for $\neg,\wedge,\ominus$
%  are straightforward. The case for atomic predicates $p$ follows from
%  the definition of restriction.
%  
%  The remaining case is that of formulas $\Diamond\phi\in\modf{\ell}$.
%  By definition of rank, $\phi\in\modf{\ell-1}$, so by induction
%  $\phi$ is $(\ell-1)$-local.
%  Let $\CA$ be some model with domain $A$ and $a\in A$. We need to
%  show $\phi_\CA(a) = \phi_\rnbhood{\CA}{\ell}{a}(a)$.
%  For any $a'\in A$ with $R^\CA(a,a') > 0$,
%  $\nbhood{\ell-1}{a'} \subseteq \nbhood{\ell}{a}$, so by locality of
%  $\phi$:
%  \[ \phi_\CA(a') = \phi_\rnbhood{\CA}{\ell-1}{a'}(a')
%   = \phi_\rnbhood{(\rnbhood{\CA}{\ell}{a})}{\ell-1}{a'}(a')
%   = \phi_\rnbhood{\CA}{\ell}{a}(a'). \]
%  Finally, since the $a'\in A$ with $R^\CA(a,a') = 0$ do not change the
%  supremum,
%  \begin{align*}
%    (\Diamond\phi_\CA)(a)
%    & = \bigvee_{a'\in A} R^\CA(a,a')\land\phi_\CA(a') \\
%    & = \bigvee_{a'\in A}
%    R^\rnbhood{\CA}{\ell}{a}(a,a')\land\phi_\rnbhood{\CA}{\ell}{a}(a')
%      = \Diamond\phi_\rnbhood{\CA}{\ell}{a}(a). \\
%  \end{align*}
%\end{proof}
\noindent To establish the desired locality result, we employ
Ehrenfeucht-Fra\"iss\'e games, introduced next. We phrase
Ehrenfeucht-Fra\"iss\'e equivalence in terms of a pseudometric, in
line with our treatment of behavioural distance, as this is the right
way of measuring equivalence w.r.t.\ fuzzy FOL; in the further
technical development, we will actually need only the case with
deviation $\epsilon=0$.
\begin{defn}
  \label{def:ef-game}
  Let $\CA,\CB$ be fuzzy relational models, and let
  $\bar a_0$ and $\bar b_0$ be vectors of equal length over $A$ and
  $B$ respectively. The \emph{$\epsilon$-Ehrenfeucht-Fra\"iss\'e game
  for $\CA,\bar a_0$ and $\CB,\bar b_0$} played by $S$
  (\emph{spoiler}) and $D$ (\emph{duplicator}) is given as follows.
  \begin{itemize}
  \item \emph{Configurations:} pairs $(\bar a,\bar b)$ of vectors
  $\bar a$ over $A$ and $\bar b$ over $B$.
  \item \emph{Initial configuration:} $(\bar a_0,\bar b_0)$.
  \item \emph{Moves:} Player $S$ needs to pick a new state in one of
    the models, say $a\in A$, and then $D$ needs to pick a state in
    the other model, say $b\in B$. The new configuration is then
    $(\bar a a,\bar b b)$.
  \item \emph{Winning condition:} Any player who needs to move but
    cannot, loses. Player $D$ additionally needs to maintain the
    condition that $(\bar a,\bar b)$ is a \emph{partial isomorphism up
      to~$\epsilon$}: For all $0\le i,j\le n$:
    \begin{itemize}
      \item $a_i=a_j\iff b_i=b_j$
      \item $|p^\CA(a_i)-p^\CB(b_i)|\le\epsilon$ for all $p\in\At$
      \item $|R^\CA(a_i,a_j)-R^\CB(b_i,b_j)|\le\epsilon$.
    \end{itemize}
  \end{itemize}
  Here, we need only the \emph{$n$-round
    $\epsilon$-Ehrenfeucht-Fra\"iss\'e game}, which as the name
  indicates is played for at most~$n$ rounds, and $D$ wins after $n$
  rounds have been played.
\end{defn}
\noindent In analogy to the classical setup, fuzzy FOL is invariant
under Ehrenfeucht-Fra\"iss\'e equivalence in the sense that formula
evaluation is non-expansive:
\begin{lem}[Ehrenfeucht-Fra\"iss\'e invariance]
  \label{lem:ef-inv-fol}
  Let $\CA,\CB$ be fuzzy relational models and $\bar a_0,\bar b_0$
  vectors of length $m$ over~$A$ and $B$, respectively.  If $D$ wins
  the $n$-round $\epsilon$-Ehrenfeucht-Fra\"iss\'e game on
  $\bar a_0,\bar b_0$, then for every first-order formula $\phi$ with
  at most $m$ free variables and $\qr(\phi)\le n$,
  \begin{equation*}
    |\phi(\bar a_0)-\phi(\bar b_0)|\le\epsilon.
  \end{equation*}  
\end{lem}

\begin{lem}
  \label{lem:bisim-inv-local}
  Let $\phi$ be a bisimulation-invariant formula of fuzzy FOL with
  quantifier rank~$\qr(\phi)\le n$. Then $\phi$ is $k$-local for
  $k = 3^n$.
\end{lem}
\begin{proof}[Proof (sketch).]
  Let $\CA$ be a model, $a_0\in A$.
  %By Lemma~\ref{lem:ef-inv-fol}, it
  %suffices to show that $D$ wins the $n$-round
  %$0$-Ehrenfeucht-Fra\"iss\'e game for $\CA,a_0$ and
  %$\rnbhood{\CA}{k}{a_0},a_0$.
  Define models $\CB$ and $\CC$ by
  extending both $\CA$ and $\rnbhood{\CA}{k}{a_0}$ by $n$ disjoint
  copies of both $\CA$ and $\rnbhood{\CA}{k}{a_0}$ each.  By
  Lemmas~\ref{lem:bisim-inv-disjoint-union} and~\ref{lem:ef-inv-fol},
  it suffices to show that $D$ wins
  the $0$-Ehrenfeucht-Fra\"iss\'e game for $\CB,a_0$ and
  $\CC,a_0$. Indeed, $D$ wins by maintaining the following invariant,
  where we put $k_i = 3^{n-i}$ for $0\le i\le n$:
  \begin{quote}
    If $(\bar b,\bar c) = ((b_0,\dots,b_i),(c_0,\dots,c_i))$ is the
    current configuration, then there is an isomorphism between
    $\rnbhood{\CB}{k_i}{\bar b}$ and $\rnbhood{\CC}{k_i}{\bar c}$
    mapping each $b_j$ to $c_j$. \qedhere
  \end{quote}
  % The invariant can be maintained as follows: whenever $S$ picks, in
  % round $i$, a new state $a$ that has Gaifman distance of at most
  % $2k_i$ from the previous states, $D$ can pick his reply according to
  % this isomorphism. Otherwise, $D$ can pick one of the disjoint copies
  % of $\CA$ or $\rnbhood{\CA}{k}{a_0}$ (depending on where $a$ lies)
  % that has not been played to so far and reply with $a$ in that copy.
\end{proof}

\section{A Fuzzy van Benthem Theorem}\label{sec:main}

\noindent It remains only to establish the implication from locality
and bisimulation-invariance to finite-depth bisimulation invariance,
using a standard unravelling construction, to finish the proof of our
main result. 

\begin{defn}
  The \emph{unravelling} $\CA^\ast$ of a model~$\CA$ is the model with
  set $A^+$ (non-empty lists over $A$) of states and
  \begin{gather*}
    p^{\CA^\ast}(\wbar{a}) = p^\CA(\pi(\wbar{a})), \\
    R^{\CA^\ast}(\wbar{a},\wbar{a}a) = R^\CA(\pi(\wbar{a}),a),
  \end{gather*}
  for $\wbar{a}\in A^+, a\in A$, where $\pi\colon A^+ \to A$ projects
  to the last element and all other values of $R^{\CA^\ast}$ are~$0$.
\end{defn}

\begin{lem}
  \label{lem:unravelling-bisim}
  For any model $\CA$ and $a_0\in A$, $D$ wins the $0$-bisimulation game
  for $\CA,a_0$ and $\CA^\ast,a_0$.
\end{lem}
\noindent The following lemma then completes the last step in our
program as laid out in Section~\ref{sec:bisimulation-games}.
\begin{lem}
  \label{lem:local-k-bisim-inv}
  Let $\phi$ be bisimulation-invariant and $k$-local.  Then~$\phi$ is
  depth-$(k+1)$ bisimulation-invariant.
\end{lem}
\begin{proof}[Proof (sketch)]
  Use locality and unravelling (Lemma~\ref{lem:unravelling-bisim}) to
  reduce to tree models of depth~$k$, and then exploit that in such
  models, winning the depth-$(k+1)$ $\epsilon$-bisimulation game
  entails winning the unrestriced game.
\end{proof}
\noindent We finally state our main result:
\begin{thm}[Fuzzy van Benthem theorem]
  \label{thm:benthem-rosen}
  Let~$\phi$ be a formula of fuzzy FOL with one free variable and
  $\qr(\phi)=n$. If~$\phi$ is bisimulation-invariant, then $\phi$ can
  by approximated by fuzzy modal formulas of rank at most $3^n+1$,
  uniformly over all models; that is: For every $\epsilon>0$ there
  exists a fuzzy modal formula $\phi_\epsilon$ such that for every
  fuzzy relational model~$\CA$ and every $a\in\CA$,
  $|\phi(a)-\phi_\epsilon(a)|\le\epsilon$.
\end{thm}
\begin{proof}[Proof (sketch)]
  By Lemmas~\ref{lem:bisim-inv-local} and~\ref{lem:local-k-bisim-inv},
  $\phi$ is depth-$k$ bisimulation-invariant for $k = 3^n + 1$.  By
  Theorem~\ref{thm:modal-approx}, $\phi$ can be modally approximated
  in rank~$k$ on the model~$\CF$ constructed from the final chain in
  Section~\ref{sec:final-chain}. The claim then follows by
  Lemma~\ref{lem:bisim-inv-projection}.
\end{proof}
\begin{rem}
  \label{rem:partial-unravelling}
  We leave the Rosen version of the characterization theorem, i.e.\
  whether Theorem~\ref{thm:benthem-rosen} remains true over finite
  models, as an open problem. As in the classical case, the
  unravelling construction is easily made to preserve finite models by
  using \emph{partial unravelling} up to the locality
  distance. However, the model construction from the final chain in
  Section~\ref{sec:final-chain} and in fact already the stages of the
  final chain are infinite, so cannot be used in this version. We thus
  do obtain a local version of the Rosen theorem, stating that on a
  fixed finite model, every first-order formula that is
  bisimulation-invariant over finite models can be approximated by
  modal formulas. However, it is unclear whether the approximation
  then works uniformly over models, as in
  Theorem~\ref{thm:benthem-rosen}.
\end{rem}

\section{Conclusions}

We have established a fuzzy analogue of the classical van Benthem
theorem: Every fuzzy first-order formula that is
bisimulation-invariant in the sense that its evaluation map is
non-expansive w.r.t.\ a natural notion of behavioural distance can be
approximated by fuzzy modal formulas. To our knowledge this is the
first modal characterization result of this type for any multi-valued
modal logic. We do point out that we leave a nagging open problem: We
currently do not know whether the result can be sharpened to claim
that every bisimulation-invariant fuzzy first-order formula is in fact
\emph{equivalent} to a fuzzy modal formula. This contrasts with the
actual technical core of our argument: The key step in our proof is to
show that \emph{every} state property that is non-expansive w.r.t.\
\emph{depth-$n$} behavioural distance can be approximated, uniformly
across models, by fuzzy modal formulas of rank~$n$, a result that
certainly cannot be improved to on-the-nose modal expressibility.

Further issues for future research include the question whether our
main result has a Rosen variant, i.e.\ holds also over finite models,
and coverage of other semantics of the propositional operators, in
particular \L{}ukasiewicz logic. We also aim to extend the modal
characterization theorem to further multi-valued logics, such as
$[0,1]$-valued probabilistic modal logics~\cite{BreugelWorrell05},
ideally at a coalgebraic level of generality.

%% Bibliography
\bibliographystyle{myabbrv}
\bibliography{coalgml}

%% Appendix
\newpage
\appendix

\section{Omitted Proofs}

\subsubsection*{Proof of Lemma~\ref{lem:diamond-nonexp}}
Let $\supnorm{f-g}\le\epsilon$. We need to show
$\supnorm{\Diamond f-\Diamond g}\le\epsilon$, so let $a\in A$ and we
need to show $|(\Diamond f)(a) - (\Diamond g)(a)|\le\epsilon$.

Let $a'\in A$. Then $|f(a')-g(a')|\le\epsilon$ by assumption. Thus,
also $|R^\CA(a,a')\land f(a') - R^\CA(a,a') \land
g(a')|\le\epsilon$.
Now we may take the supremum over all $a'\in A$ and obtain
$|(\Diamond f)(a) - (\Diamond g)(a)|\le\epsilon$, as desired. \qed

\subsubsection*{Proof of Lemma~\ref{lem:game-transitive}}
A winning strategy for~$D$ consists in following the two existing
winning strategies in parallel. In detail,~$D$ maintains the following
invariant on configurations $(a,c)$: There exists a state $b$ such
that $(a,b)$ and $(b,c)$ are winning positions for~$D$ in the
$\epsilon$-game on $\CA,\CB$ and in the $\delta$-game on
$\CB,\mathcal{C}$, respectively.  By the triangle inequality, $(a,c)$
then satisfies the winning condition. It remains to show that~$D$ can
maintain the invariant. Suppose that~$S$ makes a move from $a$ to some
$a'$ (the case where $S$ moves in~$\mathcal{C}$ is entirely
symmetric). Let $b'$ be $D$'s reply to~$a'$ in the first game. Then
$b'$ is a valid move for $S$ in the second game, since
$R(b,b') \ge R(a,a')-\epsilon > \epsilon+\delta-\epsilon = \delta$.
So~$D$ has a reply $c'$ to $b'$ as a move by~$S$ in the second game.
and $c'$ is also a valid move for $D$ in the $(\epsilon+\delta)$-game,
since $R(c,c') \ge R(b,b')-\delta \ge R(a,a')-(\epsilon+\delta)$. Of
course, the new configuration $(a',c')$ then satisfies the invariant,
as witnessed by~$b'$. \qed

\subsubsection*{Proof of Lemma~\ref{lem:pseudometrics}}
Let $\CA$ be a model. For $a\in A$, $D$ wins the $0$-bisimulation game
on $a$, $a$ by copying every move of $S$. Thus $d^G_n(a,a) = 0$.  The
symmetry of $d^G$ follows from that of
Definition~\ref{def:bisimulation-game}.  The triangle inequality
follows from Lemma~\ref{lem:game-transitive}. \qed

\subsubsection*{Proof of Lemma~\ref{lem:metrics}}
A winning strategy for $D$ in the $\epsilon$-bisimulation game wins
also the depth-$n$ $\epsilon$-bisimulation game, showing that
$d^G_n(a,b)\le d^G(a,b)$; the other inequality is shown in the same
way. \qed

\subsubsection*{Proof of Lemma~\ref{lem:mor-zero-dist}}
Let $\epsilon>0$; we show that $D$ wins the $\epsilon$-bisimulation
game on $a$ and $f(a)$. The winning strategy is given by maintaining
the invariant that configurations are of the form $(a',f(a'))$ with
$a'\in A$. This invariant holds initially, and guarantees that~$D$
wins because $f$ preserves truth values of propositional atoms. It
remains to show that~$D$ can maintain the invariant. So suppose that
from a configuration $(a',f(a'))$, $S$ moves in $\CA$ along an edge
$R^\CA(a',a'')>\epsilon$. By the definition of bounded morphisms,
$R^\CB(f(a'),f(a''))\ge R^\CA(a',a'')$, so~$D$ can answer with
$f(a'')$, maintaining the invariant. The remaining case is that $S$
moves within~$\CB$ to some state $b$ such that
$R^\CB(f(a'),b)>\epsilon$. Since
$R^\CB(f(a'),b)=\bigvee_{f(a'')=b}R^\CA(a',a'')$, there is $a''\in A$
such that $f(a'')=b$ and $R^\CA(a',a'')\ge R^\CB(f(a'),b)-\epsilon$;
then $D$ can move to $a''$, maintaining the invariant. \qed

\subsubsection*{Full Proof of Lemma~\ref{lem:bisim-inv-modal}}

Suppose that $D$ wins the $\epsilon$-bisimulation-game for $(a,b)$.
We show that $|\phi(a)-\phi(b)|\le\epsilon$ for all $\phi\in\modf{n}$
by induction on $\phi$. The case for $c\in\Rat$ is trivial and the
case for propositional atoms $p$ follows from the winning condition.
  
The inductive cases $\phi\ominus c$, $\neg\phi$, and $\phi\land\psi$
are as follows:
\begin{align*}
  |(\phi\ominus c)(a) - (\phi\ominus c)(b)|
  & = |(\phi(a)-c)\land 0 - (\phi(b)-c)\land 0| \\
  & \le |\phi(a)-\phi(b)| \le \epsilon \\
  |(\neg\phi)(a)-(\neg\phi)(b)|
  & = |(1-\phi(a))-(1-\phi(b))|\\
  & = |\phi(b)-\phi(a)| \le \epsilon \\
  |(\phi\land\psi)(a) - (\phi\land\psi)(b)|
  & = |\phi(a)\land\psi(a) - \phi(b)\land\psi(b)| \\
  & \le |\phi(a)-\phi(b)|\lor|\psi(a)-\psi(b)|\\& \le \epsilon,
\end{align*}
in each case using the inductive hypothesis in the last step.
  
Finally, we treat the case for the modality~$\Diamond$: By symmetry,
it suffices to show that
$(\Diamond\phi)(b) \ge (\Diamond\phi)(a)-\epsilon$. If
$(\Diamond\phi)(a)\le\epsilon$, then this follows immediately, so
assume $(\Diamond\phi)(a)>\epsilon$.  Let $\delta>0$; then there
is~$a'$ such that $R(a,a')>\epsilon$ and
$(\Diamond\phi)(a)-(R(a,a')\land\phi(a'))<\delta$.  Let $b'$ be $D$'s
winning answer to $S$'s move $a'$.  Then by induction,
$|\phi(a')-\phi(b')|\le\epsilon$, and moreover
$R(b,b')\ge R(a,a')-\epsilon$ by the rules of the game. Thus,
\begin{multline*}
  (\Diamond\phi)(b)\ge R(b,b')\land\phi(b')\ge
  (R(a,a')-\epsilon)\land(\phi(a')-\epsilon)\\
  =(R(a,a')\land\phi(a'))-\epsilon >(\Diamond\phi)(a)-\epsilon-\delta.
\end{multline*}
Since this holds for every $\delta$, it follows that
$(\Diamond\phi))(b)\ge(\Diamond\phi))(a)-\epsilon$. \qed
  %We show symmetrically that
  %$(\Diamond\phi)(a)\ge(\Diamond\phi)(b)-\epsilon$, that is,
  %$|(\Diamond\phi)(a)-(\Diamond\phi)(b)|\le\epsilon$ as required.

\subsubsection*{Proof Details for Theorem~\ref{thm:modal-approx}}

We elaborate details for the base case $n=0$.  For
Item~\ref{item:metrics-equal}, we show that all distances are~$0$ for
all pairs $(a,b)\in A\times A$. We have $d^G_0(a,b) = 0$ by
Remark~\ref{rem:zero-game}. To see that $d^L_0(a,b) = 0$, just note
that the only modal formulas of rank $0$ are the constants $c\in\Rat$
and Boolean combinations thereof. Finally, $d^K_0(a,b) = 0$ by
definition.

Item~\ref{item:tot-bounded} then follows trivially from the fact that
$d_0$ vanishes.

For Item~\ref{item:modal-approx}, since $d_0$ is the zero
pseudometric, the space $\nonexp{(A,d_0)}{([0,1],d_e)}$ is just the
space of constant functions. The claim then follows, since the modal
formulas correspond to the rational-valued constant functions, and
$[0,1]\cap\Rat$ is a dense subset of $[0,1]$. \qed

\subsubsection*{Full Proof of Lemma~\ref{lem:metrics-equal-LK}}

Let $a,b\in A$, put $F:=\nonexp{(A,d_{n-1})}{([0,1],d_e)}$, and define
the map
\[ H\colon (F,d_\infty) \to ([0,1],d_e),
   f \mapsto |(\Diamond f)(a) - (\Diamond f)(b)|. \]
This map is continuous because of Lemma~\ref{lem:diamond-nonexp}, and
using that function evaluation, subtraction, and taking the absolute
value are continuous operations.

By the induction hypothesis, $\modf{n-1}$ is dense in $F$, so
$H[\modf{n-1}]$ is also dense in $H[F]$, and $\bigvee H[\modf{n-1}] =
\bigvee H[F]$. Now:
\begin{align*}
  d^K_n(a,b) = & \bigvee_{p \in \textrm{At}}|p(a)-p(b)| \lor
  \bigvee_{\mathclap{f: \nonexp{(A,d_{n-1})}{([0,1],d_e)}}}
  | (\Diamond f)(a) - (\Diamond f)(b) | \\
  = & \bigvee_{p \in \textrm{At}}|p(a)-p(b)| \lor
  \bigvee_{\mathclap{\rk\phi\le n-1}}
  | (\Diamond \phi)(a) - (\Diamond \phi)(b) | \\
  = & \bigvee_{\mathclap{\rk\phi\le n}}|\phi(a) - \phi(b)|
  = d^L_n(a,b).
\end{align*}
In the second to last step, we have used the fact that $\modf{n}$ is
the set of Boolean combinations of formulas $p\in\At$ and
$\Diamond\phi$ with $\phi\in\modf{n-1}$. So the ``$\le$'' part of this
step is clearly true, and the ``$\ge$'' part also holds by a simple
induction over the Boolean combinations, using that these are
non-expansive, as shown by the same calculations as in the proof of
Lemma~\ref{lem:bisim-inv-modal}. \qed

\subsubsection*{Proof Details for Lemma~\ref{lem:metrics-equal-GK}}

  Let $a,b \in A$ and $d^L_n(a,b) \le \epsilon$. We need to show
  $d^G_n(a,b) \le \epsilon$, so it suffices to show that $D$ wins the
  depth-$n$ $(\epsilon+\delta)$-bisimulation game for every $\delta > 0$.
  The winning condition for the configuration $(a,b)$ is satisfied by
  assumption, so now suppose $S$ makes the first move from $a$ to $a'$.
  We now consider the function
  \begin{equation*}
    f\colon \nonexp{(A,d_{n-1})}{([0,1],d_e)},
    f(b') = R(a,a') \ominus d_{n-1}(a',b')
  \end{equation*}
  
  $f$ is non-expansive because it is composed from non-expansive
  functions: the map $x \mapsto c\ominus x$ is non-expansive for every
  $c\in\mathbb{R}$, and the map $b' \mapsto d_{n-1}(a',b')$ is
  non-expansive by the triangle inequality.
  
  Now,
  \begin{equation*}
    |(\Diamond f)(a) - (\Diamond f)(b)| \le d^K_n(a,b) = d^L_n(a,b)
    \le \epsilon
  \end{equation*}
  by Lemma~\ref{lem:metrics-equal-LK}, so, as $R(a,a') > \epsilon$ by
  the rules of the game,
  \begin{equation*}
    (\Diamond f)(b) \ge (\Diamond f)(a) - \epsilon
    \ge R(a,a') \land f(a') - \epsilon = R(a,a') - \epsilon > 0.
  \end{equation*}
  This means there exists some $b' \in A$ such that $f(b') > 0$ and
  \begin{align*}
    R(a,a')-\epsilon \le (\Diamond f)(b) \le R(b,b') \land f(b') +
    \tfrac{\delta}{2}.
  \end{align*}
  Rearranging this, we get
  \begin{align*}
  \epsilon+\tfrac{\delta}{2}
  & \ge R(a,a') - (R(b,b')\land f(b')) \\
  & = (R(a,a')-R(b,b')) \lor (R(a,a')-f(b')).
  \end{align*}
  So first, $R(b,b')\ge R(a,a')-(\epsilon+\frac{\delta}{2})$, which
  means that $b'$ is a legal reply for $D$. And second, since
  $f(b')>0$, 
  \begin{align*}
    \epsilon+\tfrac{\delta}{2} & \ge R(a,a')-f(b') \\
    & = R(a,a')-(R(a,a')-d_{n-1}(a',b')) = d_{n-1}(a',b').
  \end{align*}
  So the configuration reached after the first round of the game is
  $(a',b')$, and $D$ has a winning strategy for the
  $(\epsilon+\frac{\delta}{2}+\gamma)$-game for every $\gamma > 0$, in
  particular $D$ wins the $(\epsilon+\delta)$-game.
  An analogous argument can be used if $S$ makes a move from $b$ to
  some $b'$ instead. \qed

\subsubsection*{Proof of Lemma~\ref{lem:arzela-ascoli}}

  Put $F := \nonexp{(X,d_1)}{(Y,d_2)}$.  Let $\epsilon > 0$. We need
  to find a finite cover of $\nonexp{(X,d_1)}{(Y,d_2)}$ by sets of
  diameter at most~$\epsilon$.

Since $(X,d_1)$ and $(Y,d_2)$ are totally bounded, there exist finite
$\frac{\epsilon}{4}$-covers
$x_1,\dots,x_n$ of $X$ and $y_1,\dots,y_k$ of $Y$.

Now consider the set $\Phi$ of functions $\rho\colon \{1,\dots,n\} \to
\{1,\dots,k\}$, and for every $\rho \in \Phi$ let
\begin{equation*}
  F_\rho = \{ f \in F \mid f(x_i) \in
  \ball{d_2}{\frac{\epsilon}{4}}{y_{\rho(i)}} \text{ for all } 1 \le i \le n
  \}.
\end{equation*}
Then clearly $F = \bigcup_{\rho \in \Phi} F_\rho$, so it remains to
show that each $F_\rho$ has diameter at most $\epsilon$.

So let $f,g \in F_\rho$, and consider some $x \in X$. There exists
some $i$ such that $x \in \ball{d_1}{\frac{\epsilon}{4}}{x_i}$. Now,
by non-expansivity of $f$ and $g$, and the definition of $F_\rho$,
\begin{align*}
  d_2(f(x),g(x)) & \le
  d_2(f(x),f(x_i)) + d_2(f(x_i),y_{\rho(i)}) \\
  & + d_2(y_{\rho(i)},g(x_i)) + d_2(g(x_i),g(x)) \\
  & \le d_1(x,x_i) + \tfrac{\epsilon}{4} + \tfrac{\epsilon}{4} +
  d_1(x_i,x) \\
  & \le \tfrac{\epsilon}{4} + \tfrac{\epsilon}{4} +
  \tfrac{\epsilon}{4} + \tfrac{\epsilon}{4} = \epsilon.
\end{align*}
As $f$, $g$, and $x$ were chosen arbitrarily, the diameter of $F_\rho$ is
at most $\epsilon$. \qed

\subsubsection*{Proof Details for Lemma~\ref{lem:tot-bounded}}
%The space $(A,d_{n-1})$ is totally bounded by the induction
%hypothesis, and $([0,1],d_e)$ is compact and thus totally bounded,
%so by Lemma~\ref{lem:arzela-ascoli}, the space $F :=
%\nonexp{(A,d_{n-1})}{([0,1],d_e)}$ is totally bounded as well.

We show that the map
\begin{align*}
I\colon A &\to [0,1]^{k+m}\\
  a& \mapsto
(p_1(a),\dots,p_k(a),(\Diamond\phi_1)(a),\dots,(\Diamond\phi_m)(a))
\end{align*}
is indeed an $\frac{\epsilon}{3}$-isometry.
Let $a,b\in A$. We need to show that
\begin{equation*}
|d_n(a,b) - \supnorm{I(a) - I(b)}| \le \tfrac{\epsilon}{3}.
\end{equation*}
Let $f\in F$ and choose $\phi_i$ such that $\supnorm{f -
\phi_i} \le \frac{\epsilon}{6}$. Then also $\supnorm{\Diamond f
-\Diamond \phi_i} \le \frac{\epsilon}{6}$, by
Lemma~\ref{lem:diamond-nonexp}.
%Now,
%\begin{align*}
%  |(\Diamond f)(a)-(\Diamond f)(b)| & \le
%  |(\Diamond f)(a)-(\Diamond\phi_i)(a)| +
%  |(\Diamond\phi_i)(a)-(\Diamond\phi_i)(b)| \\
%  & + |(\Diamond\phi_i)(b)-(\Diamond f)(b)| \\
%  & \le |(\Diamond\phi_i)(a)-(\Diamond\phi_i)(b)| + \tfrac{\epsilon}{3} \\
%\end{align*}
%and similarly $ |(\Diamond\phi_i)(x)-(\Diamond\phi_i)(y)| \le
%|(\Diamond f)(x)-(\Diamond f)(y)| + \frac{\epsilon}{3}$.
By the triangle inequality, it follows that
\begin{equation*}
  \big| |(\Diamond f)(a) - (\Diamond f)(b)| -
    |(\Diamond\phi_i)(a) - (\Diamond\phi_i)(b)| \big|
    \le \tfrac{\epsilon}{3},
\end{equation*}
so, taking the supremum over all $f\in F$:
\begin{equation*}
  \Big| \bigvee_{f\in F} |(\Diamond f)(a) - (\Diamond f)(b)| -
    \bigvee_{i\le m} |(\Diamond\phi_i)(a) - (\Diamond\phi_i)(b)| \Big|
    \le \tfrac{\epsilon}{3}.
\end{equation*}
Recall from Assumption~\ref{ass:fin-atoms} that
$\At = \{p_1,\dots,p_k\}$ is finite; then
\begin{equation*}
\Big| d^K_n(a,b) - \bigvee_{i\le k} |p_i(a)-p_i(b)| \lor
\bigvee_{i\le m} |(\Diamond\phi_i)(a) - (\Diamond\phi_i)(b)| \Big|
\le \tfrac{\epsilon}{3}, 
\end{equation*}
which is what we needed to show.

It remains to give a finite $\epsilon$-cover of $(A,d_n)$.
As $[0,1]^{k+m}$ is compact under the supremum metric, it has a finite
$\frac{\epsilon}{3}$-cover $v_1,\dots,v_p$.
Then the pre-image
of each ball $B_\frac{\epsilon}{3}(v_i)$ has diameter at most
$\epsilon$: for any $a,b \in I^{-1}[B_\frac{\epsilon}{3}(v_i)]$,
\begin{align*}
d_n(a,b) & \le \supnorm{I(a) - I(b)} + \tfrac{\epsilon}{3} \\
& \le \supnorm{I(a) - v_i} + \supnorm{v_i - I(b)} + \tfrac{\epsilon}{3}
\le \epsilon.
\end{align*}
So a finite $\epsilon$-cover of $(A,d_n)$ arises by taking one element
from each (non-empty) $I^{-1}[B_\frac{\epsilon}{3}(v_i)]$. \qed

\subsubsection*{Proof of Lemma~\ref{lem:stone-weierstrass}}
Let $f\in F$ and $\epsilon > 0$. We need to find some $f_\epsilon \in
L$ such that $\supnorm{f-f_\epsilon} \le \epsilon$. 

By total boundedness, there exists an $\frac{\epsilon}{4}$-cover
$x_1,\dots,x_n$ of $(X,d)$.  By assumption, for every
$i,j \in \{1,\dots,n\}$ there exists some $f_{ij} \in L$ such that
$|f(x_i)-f_{ij}(x_i)| \le \frac{\epsilon}{2}$ and
$|f(x_j)-f_{ij}(x_j)| \le \frac{\epsilon}{2}$.  Now define
$f_\epsilon = \bigvee_{i\le n} \bigwedge_{j\le n} f_{ij} \in L$.
Then, for any $x \in X$ there exists some $k$ such that
$d(x_k,x) \le \frac{\epsilon}{4}$ and thus:
\begin{align*}
  f_\epsilon(x) & = \bigvee_{i\le n} \bigwedge_{j\le n} f_{ij}(x)
                \le \bigvee_{i\le n} f_{ik}(x)
                \le \bigvee_{i\le n} f_{ik}(x_k) + \tfrac{\epsilon}{4} \\
              & \le \bigvee_{i\le n} f(x_k) + \tfrac{3\epsilon}{4}
                = f(x_k) + \tfrac{3\epsilon}{4}
                \le f(x) + \epsilon,
\end{align*}
and, symmetrically:
\begin{align*}
  f_\epsilon(x) & = \bigvee_{i\le n} \bigwedge_{j\le n} f_{ij}(x)
                \ge \bigwedge_{j\le n} f_{kj}(x)
                \ge \bigwedge_{j\le n} f_{kj}(x_k) - \tfrac{\epsilon}{4} \\
              & \ge \bigwedge_{j\le n} f(x_k) - \tfrac{3\epsilon}{4}
                = f(x_k) - \tfrac{3\epsilon}{4}
                \ge f(x) - \epsilon,
\end{align*}
where we have used non-expansivity of $f$ and the $f_{ij}$ as well as
the originally assumed property of the $f_{ij}$. \qed

\subsubsection*{Proof Details for Lemma~\ref{lem:modal-approx}}

Lemma~\ref{lem:stone-weierstrass} can be applied because $(A,d_n)$ is
totally bounded by Lemma~\ref{lem:tot-bounded}, and because the set
$\modf{n}$ is clearly closed under $\land$ and $\lor$.

Given a function $f\colon\nonexp{(A,d_n)}{([0,1],d_e)}$, $a,b \in A$
and $\epsilon > 0$, we need to find $\phi\in\modf{n}$ such that
$|f(a)-\phi(a)| \le \epsilon$ and $|f(b)-\phi(b)| \le \epsilon$.

W.l.o.g. $f(a) \ge f(b)$ (otherwise we can pass to $1-f$ and negate
the resulting formula).
Now put $\Delta = f(a)-f(b)$. Then $\Delta \le d_n(a,b)$
by non-expansivity of $f$. Since $d_n = d^L_n$, there exists $\psi
\in \modf{n}$ such that $\Delta - \tfrac{\epsilon}{2} \le \psi(a) - \psi(b)$.
Let $u,v,w \in \Rat \cap [0,1]$ such that
\begin{gather*}
\ineqm{u}{\psi(b)} \\
\ineqm{v}{\Delta} \\
\ineqp{w}{f(b)}.
\end{gather*}
Put $\phi = \neg(\neg((\psi\ominus u) \land v) \ominus w)$.
Then $\phi$ approximates $f$ at $a$ and $b$:
\begin{gather*}
\ineqb{\phi(a)}{f(a)} \\
f(b)\le\phi(b)\le f(b)+\epsilon.
\end{gather*}
The detailed calculations for the above inequalities follow.
Evaluating subformulas at $a$ gives:
\begin{gather*}
  \ineqp{(\psi\ominus u)(a)}{\psi(a)-\psi(b)} \\
  \ineqm{((\psi\ominus u)\land v)(a)}{\Delta} \\
  \ineqp{(\neg((\psi\ominus u)\land v))(a)}{1-\Delta} \\
  \ineqb{(\neg((\psi\ominus u)\land v)\ominus w)(a)}{1-f(a)} \\
  \ineqb{\phi(a)}{f(a)}.
\end{gather*}
Evaluating subformulas at $b$ gives:
\begin{gather*}
  \ineqp{(\psi\ominus u)(b)}{0} \\
  \ineqp{((\psi\ominus u)\land v)(b)}{0} \\
  \ineqm{(\neg((\psi\ominus u)\land v))(b)}{1} \\
  1-f(b)-\epsilon\le
  (\neg((\psi\ominus u)\land v)\ominus w)(b)\le
  1-f(b) \\
  f(b)\le\phi(b)\le f(b)+\epsilon. 
\end{gather*}
\qed

\subsubsection*{Full Proof of Lemma~\ref{lem:bisim-inv-projection}}
We show by induction on $n$ that~$D$ wins the depth-$n$
$\epsilon$-bisimulation game for every~$\epsilon>0$. The base case is
trivial since the depth-$0$ game is an immediate win; we proceed with
the inductive step from $n$ to $n+1$. We need to show that $D$ wins
the depth-$(n+1)$ $\epsilon$-bisimulation game for $\CA,a$ and
$\CF,\pi_{n+1}(a)$.  By the explicit definition~\eqref{eq:final-cone}
of $\pi_{n+1}$, it is immediate that the winning condition holds in
the initial configuration.
  
If $S$ makes the first move from $a$ to some $a'\in A$, then~$D$ can
reply with $\pi_n(a')$, since by~\eqref{eq:final-cone},
\begin{align*}\textstyle
  R^{\CF}(\pi_{n+1}(a),\pi_n(a'))
  = \bigvee_{\pi_n(a'')=\pi_n(a')}R^\CA(a,a'') \ge R^\CA(a,a').
\end{align*}
By induction, $d^G_n(a',\pi_n(a')) = 0$, so~$D$ wins.

If instead, $S$ makes the first move from $\pi_{n+1}(a)$ to some
$y\in F$, then $R^{\CF}(\pi_{n+1}(a),y) > 0$ by the rules of the game,
so $y\in F_n$ by construction of $R^\CF$. By~\eqref{eq:final-cone},
\begin{equation*}
  R^{\CF}(\pi_{n+1}(a),y) = 
  \textstyle{\bigvee_{\pi_n(a')=y}R^\CA(a,a')}).
\end{equation*}
Thus, $D$ can pick $a'\in A$ with $\pi_n(a')=y$ such that
\[ R^\CA(a,a') \ge R^{\CF}(\pi_{n+1}(a),y) -\epsilon. \]
By induction, $d^G_n(a',y) = 0$, so $D$ wins. \qed

\subsubsection*{Proof of Lemma~\ref{lem:nbhood-bisim}}

$D$ wins the game by copying every move $S$ makes. By the definition
of $R^\rnbhood{\CA}{k}{a_0}$ and the $p^\rnbhood{\CA}{k}{a_0}$ it is
clear that such a strategy is winning as long as the game never leaves
the neighbourhood $\nbhood{k}{a_0}$. By the rules of the game, $S$ can
only ever move along positive edges of the model, so if the
configuration after round $i$ is $(a_i,a_i)$, it must hold that
$D(a_0,a_i)\le i$ and therefore $a_i\in\nbhood{k}{a_0}$. \qed

\subsubsection*{Proof of Lemma~\ref{lem:ef-inv-fol}}

  Induction on $\phi$.
  The cases for equality, propositional atoms and the relation symbol
  $R$ follow from the winning condition. The Boolean cases are proved
  just as in Lemma~\ref{lem:bisim-inv-modal}. The remaining case is
  that of existential quantification:

  Let $(\bar a,\bar b)$ be the current configuration.
  Now let $\delta>0$, let $a$ be such that
  \begin{equation*}
    (\exists x.\,\phi)(\bar a)- \phi(\bar a a)< \delta,
  \end{equation*}
  and let $b$ be $D$'s winning answer to $S$'s move $a$. Then by
  induction, $|\phi(\bar a a)-\phi(\bar b b)|\le \epsilon$. Thus,
  \begin{equation*}
    (\exists x.\,\phi)(\bar b)\ge\phi(\bar b b)\ge\phi(\bar a a)-\epsilon
    >(\exists x.\,\phi)(\bar a)-\epsilon-\delta.
  \end{equation*}
  Since $\delta>0$ was arbitrary, it follows that 
  \begin{equation*}
    (\exists x.\,\phi)(\bar b)\ge(\exists x.\,\phi)(\bar a)-\epsilon.
  \end{equation*}
  We show symmetrically that
  $(\exists x.\,\phi)(\bar a)\ge(\exists x.\,\phi)(\bar b)-\epsilon$,
  that is,
  $|(\exists x.\,\phi)(\bar a)-(\exists x.\,\phi)(\bar b)|\le\epsilon$
  as required. \qed

\subsubsection*{Proof Details for Lemma~\ref{lem:bisim-inv-local}}

Recall that $D$ needs to maintain the following invariant:

\begin{quote}
  If $(\bar b,\bar c) = ((b_0,\dots,b_i),(c_0,\dots,c_i))$ is the
  current configuration then there is an isomorphism between
  $\rnbhood{\CB}{k_i}{\bar b}$ and $\rnbhood{\CC}{k_i}{\bar c}$
  mapping each $b_j$ to $c_j$.
\end{quote}
\noindent This invariant clearly holds at the beginning of the game: the
initial configuration is $(a_0,a_0)$, and $k_0 = k$, so the two models
in the invariant are both isomorphic to $\rnbhood{\CA}{k}{a_0}$ and
the isomorphism between them maps $a_0$ to itself.

The invariant also implies the winning condition for $D$, i.e.~that
the current configuration is a partial isomorphism up to $0$. This is
because the isomorphism from the invariant maps each $b_j$ to the
corresponding $c_j$.

It remains to show that $D$ has a way to maintain the invariant.
Suppose that $i<n$ and the current configuration is as in the
invariant.

First, suppose that $S$ picks $b\in\nbhood{2k_{i+1}}{\bar b}$. Then
$D$ picks a reply $c$ according to the isomorphism. By the triangle
inequality for Gaifman distance,
$\nbhood{k_{i+1}}{b}\subseteq\nbhood{k_i}{\bar b}$ (since $2k_{i+1} +
k_{i+1} = 3k_{i+1} = k_i$), and thus also
$\nbhood{k_{i+1}}{c}\subseteq\nbhood{k_i}{\bar c}$ by isomorphism.
This implies that the domain $\nbhood{k_{i+1}}{\bar bb}$ and range
$\nbhood{k_{i+1}}{\bar cc}$ of the presumptive new isomorphism are
contained in the domain and range of the old one. So the new
isomorphism can be taken to be the restriction of the old isomorphism
to the new domain and range. The case where $S$ picks a new state
$c\in\nbhood{2k_{i+1}}{\bar c}$ is entirely symmetric.

Otherwise, suppose $S$ picks some $b$ in $\CB$ with
$b\notin\nbhood{2k_{i+1}}{\bar b}$. Then, by the triangle inequality
for Gaifman distance,
$\nbhood{k_{i+1}}{\bar b}\cap\nbhood{k_{i+1}}{b}=\emptyset$. In this
case,~$D$ picks as his reply~$c$ the copy of~$b$ in a fresh copy of
either $\CA$ or $\rnbhood{\CA}{k}{a_0}$ (i.e.~one that has not been
played to in the previous rounds). Such a fresh copy is always
available, because at most one of them gets visited in each
round. Then the radius-$k_{i+1}$ neighbourhoods of $b$ and $c$ are
isomorphic because $b$ and $c$ are the same element in isomorphic
copies of either $\CA$ or $\rnbhood{\CA}{k}{a_0}$. The
radius-$k_{i+1}$ neighbourhoods of $\bar b$ and $\bar c$ are also
isomorphic, by restriction of the old isomorphism. We thus have two
isomorphisms with disjoint domains and ranges, which we combine to
form the requested new isomorphism. Again, the case where $S$ plays in
$\CC$ instead is symmetric. \qed

\subsubsection*{Proof of Lemma~\ref{lem:unravelling-bisim}}

A winning strategy for $D$ is given by $\pi\colon A^+\to A$,
i.e.~projection to the last element. More precisely, $D$ wins by
maintaining the invariant that the current configuration is of the
form $(a,\bar a)$ with $\pi(\bar a) = a$. By definition of
$p^{\CA^\ast}$ the invariant implies the winning condition. If $S$
moves from $a$ to some $a'$, then $D$ can reply with a move from $\bar
a$ to $\bar a a'$, which is legal by definition of $R^{\CA^\ast}$. The
situation is symmetric if $S$ makes a move in $\CA^\ast$ instead. \qed

\subsubsection*{Full Proof of Lemma~\ref{lem:local-k-bisim-inv}}
Let $d_{k+1}^G(a,b)<\epsilon$; we show that
$|\phi_\CA(a)-\phi_\CB(b)|\le\epsilon$, which proves the claim. By
assumption, $D$ wins the $\epsilon$-bisimulation-game for $\CA,a$ and
$\CB,b$. By Lemmas~\ref{lem:game-transitive},~\ref{lem:nbhood-bisim}
and~\ref{lem:unravelling-bisim}, $D$ also wins the depth-$(k+1)$
$\epsilon$-bisimulation game for $\rnbhood{(\CA^\ast)}{k}{a},a$ and
$\rnbhood{(\CB^\ast)}{k}{b},b$.
  
The models $\rnbhood{(\CA^\ast)}{k}{a}$ and
$\rnbhood{(\CB^\ast)}{k}{b}$ both have the shape of trees of depth
$k$, so for every $0\le i\le k$, before the start of round $i+1$ of
the above game, the two states on either side of the current
configuration are nodes at distance $i$ from the root of their tree
(i.e.~$a$ or $b$). In particular, if round $k+1$ needs to be played,
then~$S$ has no legal move, because the current configuration consists
of two leaf nodes.
  
Using this observation, we conclude that $D$'s winning strategy for
the depth-$(k+1)$ game is in fact also a winning strategy for the
unbounded $\epsilon$-bisimulation game, so
$|\phi_\rnbhood{(\CA^\ast)}{k}{a}(a)-\phi_\rnbhood{(\CB^\ast)}{k}{b}(b)|\le\epsilon$,
by bisimulation invariance of $\phi$.
  
By locality and bisimulation invariance of $\phi$, and again
Lemma~\ref{lem:unravelling-bisim}, we have
$\phi_\rnbhood{(\CA^\ast)}{k}{a}(a) = \phi_{(\CA^\ast)}(a) =
\phi_\CA(a)$
as well as
$\phi_\rnbhood{(\CB^\ast)}{k}{b}(b) = \phi_{(\CB^\ast)}(b) =
\phi_\CB(b)$.
Thus $|\phi_\CA(a)-\phi_\CB(b)|\le\epsilon$, as claimed. \qed

\subsubsection*{Full Proof of Theorem~\ref{thm:benthem-rosen}}
By Lemmas~\ref{lem:bisim-inv-local} and~\ref{lem:local-k-bisim-inv},
$\phi$ is depth-$k$ bisimulation-invariant for $k = 3^n + 1$.  By
Theorem~\ref{thm:modal-approx}, $\phi$ can be modally approximated on
the model $\CF$ constructed from the final chain in
Section~\ref{sec:final-chain}, i.e.\ for every $\epsilon>0$ there
exists a modal formula $\phi_\epsilon$ of rank at most~$k$ such that
for every $x\in F$, $|\phi(x)-\phi_\epsilon(x)|\le\epsilon$. Now let
$\CA$ be a fuzzy relational model and $a\in A$.  By
Lemma~\ref{lem:bisim-inv-projection}, $\phi(a) = \phi(\pi_k(a))$ and
$\phi_\epsilon(a) = \phi_\epsilon(\pi_k(a))$ (where~$\pi_k$ is the
projection into the final chain), so we obtain
$|\phi(a)-\phi_\epsilon(a)|\le\epsilon$, as required. \qed

\subsubsection*{Details for Remark~\ref{rem:partial-unravelling}}

In the version of Lemma~\ref{lem:unravelling-bisim} where $\CA^\ast$
is the partial unravelling instead, $D$ wins with a similar,
but slightly more complicated invariant: the current configuration is
either of the form $(a,\bar a)$ with $\pi(\bar a) = a$ or it is a pair
of two equal states, the second being from one of the disjoint copies
of $\CA$. $D$ can maintain this invariant for the first $k+1$ rounds
just as before, and after that can copy $S$'s moves indefinitely
because the game is now played between identical models.

% %% Acknowledgments
% \begin{acks}                            %% acks environment is optional
%                                         %% contents suppressed with 'anonymous'
%   %% Commands \grantsponsor{<sponsorID>}{<name>}{<url>} and
%   %% \grantnum[<url>]{<sponsorID>}{<number>} should be used to
%   %% acknowledge financial support and will be used by metadata
%   %% extraction tools.
%   This material is based upon work supported by the
%   \grantsponsor{GS100000001}{National Science
%     Foundation}{http://dx.doi.org/10.13039/100000001} under Grant
%   No.~\grantnum{GS100000001}{nnnnnnn} and Grant
%   No.~\grantnum{GS100000001}{mmmmmmm}.  Any opinions, findings, and
%   conclusions or recommendations expressed in this material are those
%   of the author and do not necessarily reflect the views of the
%   National Science Foundation.
% \end{acks}

\end{document}